\documentclass[secthm,amsthm,pdftex]{elsarticle}
\usepackage{amsmath,amsfonts,amstext,amsthm,amssymb,latexsym,epic,eepic}
\usepackage[mathscr]{eucal}
\usepackage{multirow}

\usepackage{caption}
\usepackage{subcaption}
\usepackage{array}
\usepackage{framed}

\usepackage{lineno}
\modulolinenumbers[5]

\newtheorem{theorem}{Theorem}[section]
\newtheorem{lemma}[theorem]{Lemma}

\theoremstyle{definition}
\newtheorem{definition}{Definition}[section]
 \newtheorem{example}{Example}[section]
\theoremstyle{remark}

\numberwithin{figure}{section}

\newcommand{\PP}{\mathcal{P}}
\newcommand{\R}{\mathbb{R}}

\journal{Data \& Knowledge Engineering}

\begin{document}
 
\begin{frontmatter}

\title{Accurate and Efficient Profile Matching in Knowledge Bases\tnoteref{*}}
\tnotetext[*]{The research reported in this paper was supported by the Austrian Forschungsf\"orderungsgesellschaft (FFG) for the Bridge project ``Accurate and Efficient Profile Matching in Knowledge Bases'' (ACEPROM) (\textbf{FFG: [841284]}). The research has further been supported by the Austrian Ministry for Transport, Innovation and Technology, the Federal Ministry of Science, Research and Economy, and the Province of Upper Austria in the frame of the COMET center SCCH (\textbf{FFG: [844597]}). We are further grateful to the co-funding support by two companies \textbf{3S} and \textbf{OntoJob} in the frame of the ACEPROM project.}

\author[1]{Jorge Martinez Gil\corref{a}}
\ead{jorge.martinez-gil@scch.at}
\author[1]{Alejandra Lorena Paoletti}
\ead{lorena.paoletti@scch.at}
\author[2]{G\'{a}bor R\'{a}cz}
\ead{gabee33@gmail.com}
\author[2]{Attila Sali}
\ead{sali@renyi.hu}
\author[1]{Klaus-Dieter Schewe}
\ead{kdschewe@acm.org}
\cortext[a]{Corresponding author}

\address[1]{Software Competence Center Hagenberg, Austria}
\address[2]{Alfr\'{e}d R\'{e}nyi Institute of Mathematics, Hungary}

\begin{abstract}

A profile describes a set of properties, e.g. a set of skills a person may have, a set of skills required for a particular job, or a set of abilities a football player may have with respect to a particular team strategy. Profile matching aims to determine how well a given profile fits to a requested profile and vice versa. The approach taken in this article is grounded in a matching theory that uses filters in lattices to represent profiles, and matching values in the interval [0,1]: the higher the matching value the better is the fit. Such lattices can be derived from knowledge bases exploiting description logics to represent the knowledge about profiles. An interesting first question is, how human expertise concerning the matching can be exploited to obtain most accurate matchings. It will be shown that if a set of filters together with matching values by some human expert is given, then under some mild plausibility assumptions a {\em matching measure} can be determined such that the computed matching values preserve the relevant rankings given by the expert. A second question concerns the efficient querying of databases of profile instances. For {\em matching queries} that result in a ranked list of profile instances matching a given one it will be shown how corresponding top-$k$ queries can be evaluated on grounds of pre-computed matching values, which in turn allows the maintenance of the knowledge base to be decoupled from the maintenance of profile instances. In addition, it will be shown how the matching queries can be exploited for {\em gap queries} that determine how profile instances need to be extended in order to improve in the rankings. Finally, the theory of matching will be extended beyond the filters, which lead to a matching theory that exploits fuzzy sets or probabilistic logic with maximum entropy semantics. 

\end{abstract}


\end{frontmatter}

\section{Introduction}

A profile describes a set of properties, and profile matching is concerned with the problem to determine how well a given profile fits to a requested one. Profile matching appears in many application areas such as matching applicants for jobs to job requirements, matching system configurations to requirements specifications, matching team players to game strategies in sport, etc. 

In order to support profile matching by knowledge-based tools the first question concerns the representation of the profiles. For this one might use just sets, e.g. the set of skills of a person could contain ``knowledge of Java'', ``knowledge of parsing theory'', ``knowledge of Italian'', ``team-oriented person'', etc. In doing so, profile matching would have to determine measures for the difference of sets. Such approaches exist, but they will usually lead only to very coarse-grained matchings. For instance, ``knowledge of Java'' implies ``knowledge of object-oriented programming'', which further implies ``knowledge of programming''. Thus, at least hierarchical dependencies between the terms in a profile should be taken into account, which is the case in many taxonomies for skills. Mathematically it seems appropriate to exploit partially-ordered sets or better lattices to capture such dependencies, which justifies to consider not just sets, but filters in lattices as an appropriate way of representing profiles. 

However, other more general dependencies would still not be captured. For instance, the ``knowledge of Java'' may be linked to ``two years of experience'' and ``application programming in web services''. If such a fine-grained characterisation of skills is to be captured as well, the well established field of knowledge representation offers solutions in form of ontologies, which formally can be covered by description logics. This further permits a thorough distinction between the terminological level (aka TBox) of a knowledge base capturing abstract concepts such as ``knowledge of Java'' or ``web services'' and their relationships, and corresponding assertional knowledge (aka ABox) referring to particular instances. For example, the ABox may contain the skills of particular persons such as Lorena or Jorge linking them to ``Lorena's knowledge of Java'' or ``Jorge's knowledge of Java'', which further link to ``ten years of experience'' or ``six years of experience'', respectively. Therefore, it appears appropriate to assume that a knowledge base is used for the representation of the knowledge about abstract and concrete terms appearing in profiles.

For profile matching this raises the question, how the extended relationships could be integrated into profiles. Pragmatically, it appears justified to assume that the knowledge base is finite. For instance, for recruiting a distinction between ``$n$ years of experience'' only makes sense for positive integer values for $n$ up to some maximum. Also, a classification of application areas will only require finitely many terms. As the relationships in a knowledge base correspond semantically to binary relations, we can exploit inverse images and thus exploit concepts such as ``knowledge of Java with $n$ years of experience'' for different possible values of $n$---of course, the concept with a larger value of $n$ subsumes the concept with a smaller value''---to derive a sophisticated lattice from the knowledge base. Even the knowledge of the actual instances in the ABox can be exploited to keep the derived lattice reasonably small.

As profiles can be represented by filters in lattices, profile matching can be approached by appropriate {\em matching measures} $\mu$, i.e. we assign to each pair of filters a numeric {\em matching value} $\mu(\mathcal{F}_1,\mathcal{F}_2)$, which we can normalise to be in the interval $[0,1]$. The meaning should be the degree to which the profile represented by $\mathcal{F}_1$ fits to the profile represented by $\mathcal{F}_2$, such that a higher matching value corresponds to a better fit. Naturally, if a given profile contains everything that is requested through another profile, the given profile should be considered a perfect fit and receive a matching value of 1. This implies that the matching measure cannot be symmetric. For instance, almost everyone would satisfy the requirements for a position of ``receptionist'', but it is very unlikely that for such a position a PhD-qualified candidate would be selected. This is, because the inverted matching measure is low, i.e. the requested profile for the receptionist does not fit well to the profile of the PhD-qualified candidate, in other words, s/he is considered to be over-qualified. In this way the matching measures can be handled flexibly to capture also concepts such as over-qualification or matching by means of multiple criteria, e.g. technical skills and social skills.

\subsection{Our Contribution}

In this paper we develop a theory of profile matching, which is inspired by the problem of matching in the recruiting domain \cite{paoletti:dexa2015} but otherwise independent from a particular application domain. As argued above our theory will be based on profiles that are defined by filters in appropriate lattices \cite{martinez:medi2016}. We will show that information represented in knowledge bases using highly expressive description logics similar to DL-LITE \cite{artale:jair2009}, $\mathcal{SROIQ}$ \cite{baader:2003} or OWL-2 \cite{grau:jws2008} can be captured adequately by filters in lattices that are derived from the TBox of the knowledge base. For a matching measure we then assign weights to the elements of the lattice, where the weighting function should satisfy some normalising constraints. This will be exploited in the definition of asymmetric {\em matching measures}. We argue that every matching measure can be obtained by such weights.

While these definitions constitute a justified contribution to formalise profile matching, we will investigate how matching measures can be maintained in the light of human expertise. If bias (i.e. matching decisions that are grounded in concepts not appearing in the knowledge base) can be excluded, the question is, if human-made matchings can always be covered by an appropriate matching measure. As human experts rather use rankings than precise matching values, we formalise this problem by a notion of {\em ranking-preserving matching measure}. Our first main result is that under some mild assumptions---the satisfaction of plausibility constraints by the human-made matchings---a ranking-preserving matching measure can always be determined. The proof of this result requires to show the solvability of some linear inequalities. However, we also show that in general, not all relationships in human-made matchings can be preserved in a matching measure.

Our second main result shows how {\em matching queries} can be efficiently implemented. As matchings are usually done to determine a ranked list of matching profiles for a fixed profile---in recruiting this usually refers to the shortlisting procedure---or conversely a ranked list of profiles for which a given profile may fit, and only the $k$ top-ranked profile instances are considered to be relevant (for some user-defined integer constant $k$), the problem boils down to establish an efficient implementation for top-$k$ queries. The naive approach to first evaluate a query that determines a ranked list of profile instances, from which the ``tail'' starting with the $k+1$st element is thrown away, would obviously be highly inefficient. In addition, in our case the ranking criterium is based on the matching values, which themselves require a computation on the basis of filters, so it is desirable to minimise the number of such computations \cite{paoletti:dexa2016}. Therefore, we favour an approach that separates the profiles that are determined by the underlying lattice and thus do not change very often from the profile instances in a concrete matching query, e.g. the information about current applicants for an open position. Note that there may be several such instances associated with the same profile. Then we can create data structures around pre-computed matching values. The number of such pre-computed values is far less than the number of profile pairs, and most importantly, the data structure can exploit that only the larger ones of these values will be relevant for the matching queries. Based on these ideas we contribute an algorithm for the efficient evaluation of top-$k$ matching queries. 

In addition, we take this approach to querying further to support also {\em gap queries}, which determine for a given profile extensions that improve the rankings for the given profile with respect to possible requested profiles. In the recruitment domain such gap queries are of particular interest for job agents, who can use the results to recommend additional training to candidates that would otherwise have low chances on the job market. Furthermore, for educational institutions the results of such queries give information about the needed qualifications that should be targeted by training courses.

Finally, we investigate further extensions to the matching theory with respect to relations between the concepts in a profile that are not covered by ontologies. In particular, the presence of a particular concept in a profile may only partially imply the presence of another concept. For instance, ``knowledge of Java'' and ``knowledge of NetBeans'' may be unrelated in a knowledge base, yet with a degree (or probability) of $0.7$ the former one may imply the latter one. We therefore explore additional links between the elements of the lattice that are associated with a degree (or probability); even cycles will be permitted. We contribute an {\em enriched matching theory} by means of values associated to paths \cite{racz:foiks2016}. Regarding the values associated with the added links as probabilities suggests a different approach that exploits {\em probabilistic logic programs}, for which we exploit the maximum entropy semantics, which interprets the additional links as adding minimal additional information. We show that matching values are the result of probabilistic queries that are obtained from sentences determined by the extended links.

\subsection{Related Work}

Knowledge representation is a well established branch of Artificial Intelligence. In particular, description logics have been in the centre of interest, often as the strict formal counterpart of the more vaguely used terms ``ontology'' or ``semantic technology''. At its core a description logic can be seen as a fragment of first-order logic with unary and binary predicates with a decidable implication. A TBox captures terminological knowledge, i.e. well-formed sentences (implications) of the logic, while an ABox captures instances, i.e. assertional knowledge. Many different description logics have been developed, which differ by their expressiveness (see \cite{baader:2003} for a survey). The DL-LITE family provides a collection of description logics that appear to be mostly sufficient for our purposes, though we will also exploit partly constructs from the highly expressive description logic is $\mathcal{SROIQ}$ to highlight the relationship between knowledge representation and profile matching.

Description logics have been used in many application branches, in particular as a foundation for the semantic web \cite{klein:semweb2003} and for knowledge sharing \cite{gruber:ijhcs1995}. For the usage in the context of the semantic web the language OWL-2, which is essentially equivalent to $\mathcal{SROIQ}(D)$, has become a standard. Ontologies have also been used in the area of recruiting in connection with profile matching (see \cite{falk:is2006} for a survey). However, while it is claimed that matching accuracy can be improved \cite{mochol:bis2007}, the matching approach itself remains restricted to Boolean matching, which basically means to count how many requested skills also appear in a given profile \cite{mochol:sebiz2006}. Surprisingly, sophisticated taxonomies in the recruitment domain such as DISCO \cite{disco}, ISCO \cite{isco:2008} and ISCED \cite{isced} have not yet been properly linked with the much more powerful and elegant languages for knowledge representation.

With respect to foundations of a profile matching theory we already argued that it is not appropriate to define profiles just as sets of unrelated items, even though many distance measures for sets such as Jaccard or S{\o}rensen-Dice have proven to be useful in ecological applications \cite{levandowsky:nat1971}. The first promising attempt to take hierarchical dependencies into account was done by Popov and Jebelean \cite{popov:2013}, which defines the initial filter-based measure. However, weights are not used, only cardinalities, which correspond to the special case that all concepts are equally weighted. Our matching theory is inspired by this work, but takes the filter-based approach much farther. To our best knowledge no other approach in this direction has been tried.

With respect to the analysis of matching measures, in particular in connection with human-defined matchings it is tempting to exploit ontology learning \cite{martinez:adbis2016} or formal concept analysis (FCA) \cite{ganter:1999,ganter:jucs2004}. FCA has been exploited successfully in many areas, also in combination with ontologies \cite{cimiano:icfca2004} (capturing structure) and rough sets \cite{ganter:trs2011} (capturing vagueness). A first attempt to exploit formal concept analysis for the learning of matching measures has been reported in \cite{looser:apccm2013}. However, it turned out that starting from matching relations, the derived concept lattices still require properties to be examined that are expressed in terms of the original relation, so the ideas to exploit formal concept analysis were abandoned. 
Regarding ontology learning a first application to e-recruitment has been investigated in \cite{martinez:jikm2014}. The resulting learning algorithms can be combined with our results on the existence of ranking-preserving matching measures. Remotely related to our objective to determine matching measures that are in accordance with human-made matchings is the research on human preferences in ranking \cite{tran:kais2016} and on product ranking with user credibility \cite{zhang:kais2016}.

Top-$k$ queries have attracted a lot of attention in connection with ranking problems. Usually, they are investigated in the context of relational databases \cite{chakrabarti:tkde2004}, and the predominant problem associated with them is efficiency \cite{han:kais2016}. This is particularly the case for join queries \cite{ilyas:vldbj2004} or for the determination of a dominant query \cite{han:kais2015}. Extensions in the direction of fuzzy logic \cite{straccia:fss2012} and probabilities have also been tried \cite{theobald:vldb2004}. For our purposes here, most of the research on top-$k$ query processing is only marginally relevant, because it is not linked to the specific problem of finding the best matches. On one hand this brings with it the additional problem that the matching values need to be computed as well, but on the other hand permits a simplification, as the possible matching values do not frequently change.

Concerning enriched matching with fuzzy degrees seems at first sight to lead to the NP-complete problem of finding longest paths in a weighted directed graph, but in our case the weighting is multiplicative with values in [0,1]. This enables an interpretation using fuzzy filters \cite{hajek:1998}. For the probabilistic extension our research will be based on the probabilistic logic with maximum entropy semantics in
\cite{beierle:entropy2015,kern-isberner:ai2004}, for which sophisticated reasoning methods exist \cite{schramm:or1999}.

\subsection{Organisation of the Article}

The remainder of this article is organised as follows. In Section \ref{sec:dl} we introduce fundamentals from knowledge representation with description logics. We particularly emphasise the features in the description logics DL-LITE and $\mathcal{SROIQ}$ without adopting them completely. Actually, we leave it to the particular application domain to decide, which knowledge representation is most appropriate. However, we require that such a knowledge base gives rise to a lattice that captures the information found in profiles. We show how the roles give rise to particular subconcepts and thus can be omitted from further consideration. Thus we show how we can obtain a lattice that is needed for our matching theory from a knowledge base that captures general terminology of an application domain. For instance, we envision that on one hand for recruiting an extension of the various taxonomies such as ISCO, DISCO and ISCED perfectly makes sense even without this connection to matching, while on the other hand matching has to exploit the available knowledge sources.

In Section \ref{sec:matching} we introduce the fundamentals of our approach to profile matching. We start with profiles defined by filters in lattices and define weighted {\em matching measures} on top of them. Naturally, the lattices are derived from knowledge bases as discussed in Section \ref{sec:dl}. We further discuss the lattice of {\em matching value terms}, which symbolically characterise possible matching values.

Section \ref{sec:learning} is then dedicated to the relationship between the filter-based matching theory and matchings by human experts. That is, the section is dedicated to the problem to determine, if and how a matching measure as defined in Section \ref{sec:matching} can be obtained from human-defined matching values. For this we first derive plausibility constraints that human-made matching should fulfil in order to exclude unjustified bias. We then show that if the plausibility constraints are satisfied, weights can be defined in such a way that particular rankings based on the corresponding matching measure coincide with the human-made ones. The rankings we consider are restricted to either the same requested profile or the same given profile plus requested profiles in the same relevance class. This permits minimum updates to existing matching measures in order to establish compliance with human expertise.

In Section \ref{sec:queries} the issue of queries is addressed, which concerns the implementation of the matching theory from Section \ref{sec:matching}. First we present an approach to implement top-$k$ queries for matching, which extend naturally to matching queries under multiple criteria, e.g. capturing over-qualification. The second class of queries investigated concerns gaps, i.e. possible enlargements of profiles that lead to better rankings of a profile instance.

Section \ref{sec:probability} is dedicated to enriched matchings that exploit additional links between concepts in the knowledge base. This takes the matching theory from Section \ref{sec:matching} further. First we discuss maximum length matching, which is based on a fuzzy set interpretation of such links. Second, we discuss an interpretation in probabilistic logic with maximum entropy semantics. Naturally, for the enriched matching theory it would be desirable to address again the issues of preservation of human-defined rankings and efficient querying, but these topics are still under investigation and thus left out of this article.

Finally, we conclude the article in Section \ref{sec:schluss} with a brief summary and discussion of open questions that need to be addressed to apply our matching theory in practice.

\section{Profiles and Knowledge Bases}\label{sec:dl}

This section is dedicated to a brief introduction of our understanding of knowledge bases that form the background for our approach to profile matching. In the introduction we already outlined that we consider a profile to describe a set of properties, and that dependencies between such properties should be taken into account. Therefore, our proposal is to exploit description logics for the representation of domain knowledge. Thus, for the representation of knowledge we adopt the fundamental distinction between {\em terminological} and {\em assertional} knowledge that has been used in description logics since decades. In accordance with basic definitions about description logics \cite[pp.~17ff.]{baader:2003} for the former one we define a language, which defines the TBox of a knowledge base, while the instances define corresponding ABoxes.

A TBox consists of concepts, roles and constraints on them. The description logic used here is derived from DL-LITE \cite{artale:jair2009} and $\mathcal{SROIQ}$ \cite{baader:2003}, which we use as role models for the features that in many application domains need to be supported.

\begin{description}

\item[$\mathcal{S}$] stands for the presence of a top concept $\top$, a bottom concept $\bot$, intersection $C_1 \sqcap C_2$, union $C_1 \sqcup C_2$, and for concepts $\forall R. C$ and $\exists R. C$ (the semantics of these will be defined later).

\item[$\mathcal{R}$] stands for role chains $R_1 \circ R_2$ and role hierarchies $R_1 \sqsubseteq R_2$ (the latter ones we do not need).

\item[$\mathcal{O}$] stands for nominals, i.e. we provide individual $I_0$ and permit to use concepts of the form $\{ a \}$. Then in combination with $\mathcal{S}$ also enumerations $\{ a_1, \dots, a_n \} = \{ a_1 \} \sqcup\dots\sqcup \{a_n \}$ are enabled.

\item[$\mathcal{I}$] stands for inverse atomic roles $R_0^-$.

\item[$\mathcal{Q}$] stands for quantified cardinality restrictions $\ge m. R. C$ and $\le m. R. C$ (the semantics of these will be defined later).

\end{description}

We believe that in most application domains---definitely for job recruiting---it is advisable to exploit these features in a domain knowledge base, except role hierarchies and inverse roles. Therefore, let us assume that $C_0$, $I_0$ and $R_0$ represent not further specified sets of basic concepts, individuals and roles, respectively. Then {\em atomic concepts} $A$, {\em concepts} $C$ and {\em roles} $R$ are defined by the following grammar:
\begin{align*}
R \quad &= \quad R_0 \mid R_1 \circ R_2 \\
A \quad &= \quad C_0 \mid \top \mid \; \ge m . R . C \; (\text{with $m > 0$}) \mid \{ I_0 \} \\
C \quad &= \quad A \mid \neg C \mid C_1 \sqcap C_2 \mid C_1 \sqcup C_2 \mid \exists R . C \mid \forall R . C
\end{align*}

\begin{definition}\label{def-tbox}\rm

A {\em TBox} $\mathcal{T}$ comprises concepts $C$ and roles $R$ as defined by the grammar above plus a finite set of constraints of the form $C_1 \sqsubseteq C_2$ with concepts $C_1$ and $C_2$.

\end{definition}

Each assertion $C_1 \sqsubseteq C_2$ in a TBox $\mathcal{T}$ is called a {\em subsumption
axiom}. Note that Definition \ref{def-tbox} only permits subsumption between concepts, not between
roles, though it is possible to define more complex terminologies that also permit role
subsumption (as in $\mathcal{SROIQ}$). As usual,  we use several shortcuts: (1) $C_1 \equiv C_2$ can be used instead of $C_1 \sqsubseteq C_2 \sqsubseteq C_1$, (2) $\bot$ is a shortcut for $\neg \top$, (3) $\{ a_1 ,\dots, a_n \}$ is a shortcut for $\{ a_1 \} \sqcup\dots\sqcup \{ a_n \}$, (4) $\le m . R . C$ is a shortcut for $\neg \ge m+1 . R . C$, and (5) $= m . R . C$ is a shortcut for $\ge m . R . C \;\sqcap\; \le m . R . C$.

\begin{example}\label{bsp-tbox}

With a skill ``programming'' we would like to associate other properties such as ``years of experience'', ``application area'', and ``degree of complexity'', which defines a complex aggregate structure. In a TBox this may lead to subsumption axioms such as
\begin{gather*}
\text{programming} \sqsubseteq \exists \text{experience}.\{ 1 , \dots, 10 \} \\
\text{programming} \sqsubseteq \exists \text{area}.\{ \text{business}, \text{science}, \text{engineering} \} \; \text{and} \\
\text{programming} \sqsubseteq \exists \text{complexity}.\{ 1 , \dots, 5 \}
\end{gather*}

\end{example}

Obviously, the concepts in a TBox define a lattice $\mathcal{L}$ with $\sqcap$ and $\sqcup$ as operators for meet and join, and $\sqsubseteq$ for the partial order. For our purpose of matching we are particularly interested in named concepts, i.e. we assume that for each concept $C$ as defined by the grammar above we also have a constraint $C_0 \equiv C$ with some atomic concept name in $C_0$. Then we can identify the elements of $\mathcal{L}$ with the names in $C_0$. In particular, in order to exploit the roles as in Example \ref{bsp-tbox} we consider ``blow-up'' concepts \cite{paoletti:dexa2015} that have the form $C \sqcap \exists R. C^{\prime\prime}$, where $C$ is a concept, for which $C \sqsubseteq \exists R. C^\prime$ holds and $C^{\prime\prime} \sqsubseteq C^\prime$ holds. In particular, this becomes relevant, if $C^{\prime\prime}$ is defined by individuals, say $C^{\prime\prime} = \{ a_1 ,\dots, a_n \}$.

Using such blow-up concepts, we can express concepts such as ``programming of complexity level 4 in science with at least 3 years of experience''. We tacitly assume that for matching we exploit a lattice that is defined by a TBox exploiting blow-up concepts as well as $\sqcap$, $\sqcup$ and $\sqsubseteq$.

\begin{definition}\label{def-structure}\rm

A {\em structure} $\mathcal{S}$ for a TBox $\mathcal{T}$ consists of a non-empty set
$\mathcal{O}$ together with subsets $\mathcal{S}(C_0) \subseteq \mathcal{O}$ and $\mathcal{S}(R_0)
\subseteq \mathcal{O} \times \mathcal{O}$ for all basic concepts $R_0$ and basic roles $R_0$,
respectively, and individuals $\bar{a} \in \mathcal{O}$ for all $a \in I_0$. $\mathcal{O}$ is called the base set of the structure.

\end{definition}

We first  extend the interpretation of basic concepts and roles and to all concepts and roles as
defined by the grammar above, i.e. for each concept $C$ we define a subset $\mathcal{S}(C)
\subseteq \mathcal{O}$, and for each role $R$ we define a subset $\mathcal{S}(R) \subseteq
\mathcal{O} \times \mathcal{O}$ as follows:
\begin{gather*}
\mathcal{S}(R_1 \circ R_2) = \{ (x,z) \mid \exists y . (x,y) \in \mathcal{S}(R_1) \wedge (y,z) \in \mathcal{S}(R_2) \\
\mathcal{S}(\top) = \mathcal{O} \qquad
\mathcal{S}(\{ a \}) = \{ \bar{a} \} \qquad
\mathcal{S}(\neg C) = \mathcal{O} - \mathcal{S}(C) \\
\mathcal{S}(\ge m . R . C) = \{ x \in \mathcal{O} \mid \# \{ y \mid (x,y) \in \mathcal{S}(R) \wedge y \in \mathcal{S}(C) \} \ge m \} \\
\mathcal{S}(C_1 \sqcap C_2) = \mathcal{S}(C_1) \cap \mathcal{S}(C_2) \qquad
\mathcal{S}(C_1 \sqcup C_2) = \mathcal{S}(C_1) \cup \mathcal{S}(C_2) \\
\mathcal{S}(\exists R. C) = \{ x \in \mathcal{O} \mid (x,y) \in \mathcal{S}(R) \;\text{for some}\; y \in \mathcal{S}(C) \} \\
\mathcal{S}(\forall R. C) = \{ x \in \mathcal{O} \mid (x,y) \in \mathcal{S}(R) \Rightarrow y \in \mathcal{S}(C) \;\text{for all}\; y \}
\end{gather*}

In doing so, we can consider concepts as predicate symbols $C$ of arity 1 in first-order logic, and roles as predicate symbols of arity 2. Then the extension of the structure defines a set of ground instances of the form $C(a)$ and $R(a,b)$. An {\em ABox} for a TBox $\mathcal{T}$ is usually defined as a finite set of such ground atoms. Thus, an ABox is de facto defined by a structure, for which in addition we usually assume consistency.

\begin{definition}\label{def-abox}\rm

A structure $\mathcal{S}$ for a TBox $\mathcal{T}$ is {\em consistent} if $\mathcal{S}(C_1) \subseteq \mathcal{S}(C_2)$ holds for all assertions $C_1 \sqsubseteq C_2$ in $\mathcal{T}$.

\end{definition}

For the following we always consider a concept $C$ in a TBox as representation of abstract properties, e.g. ``knowledge of Java'', and individuals in the ABox as concrete properties such as the ``Java knowledge of Lisa''. 

\begin{definition}\label{def-profile}

Given a finite consistent structure, a {\em profile} $P$ is a subset of the base set $\mathcal{O}$. The {\em representing filter} of a profile $P$ is the filter $\mathcal{F}(P) \subseteq \mathbb{F}$ of the lattice $\mathcal{L}$ defined by the TBox with $\mathcal{F}(P) = \{ C \in \mathcal{L} \mid \exists p \in P . \; p \in \mathcal{S}(C) \}$.

\end{definition}

\section{Weighted Profile Matching Based on Lattices and Filters}\label{sec:matching}

In this section we present the formal definitions underlying our approach to profile matching. In the previous section we have seen how to obtain lattices from knowledge bases. Our theory is therefore based on such lattices $\mathcal{L}$. We have further seen that a profile, understood as a set $P$ of concept instances in an ABox, always gives rise to a representing filter $\mathcal{F}$ in the lattice $\mathcal{L}$. Therefore, our theory exploits filters in lattices. We first define the notion of a {\em matching measure} as a function defined on pairs of filters. If $\mu$ is such a matching measure and $\mathcal{F}$, $\mathcal{G}$ are filters, the $\mu(\mathcal{F},\mathcal{G})$ will be a real number in the interval $[0,1]$, which we will call a {\em matching value}. Matching measures will explot weights assigned to concepts in the lattice $\mathcal{L}$. Finally in this section we discuss a specific lattice, the lattice of {\em matching value terms} that is defined by the matching measures.

\subsection{Filter-Based Matching}

As stated above, profiles are to describe sets of properties, and profile matching should result in values that determine how well a given profile fits to a requested one, so we base our theory on lattices. Throughout this section let $(\mathcal{L}, \le)$ be a lattice. Informally, for $A, B \in \mathcal{L}$ we have $A \le B$, if the property $A$ subsumes property $B$, e.g. for skills this means that a person with skill $A$ will also have skill $B$.

\begin{definition}\label{def-filter}\rm

A {\em filter} is a non-empty subset $\mathcal{F} \subseteq \mathcal{L}$, such that for all $C, C^\prime$ with $C \le C^\prime$ whenever $C \in \mathcal{F}$ holds, then also $C^\prime \in \mathcal{F}$ holds.

\end{definition}

\noindent
We concentrate on filters $\mathcal{F}$ in order to define matching measures. 

\begin{definition}\label{def-weighting}\rm

Let $\mathbb{F} \subseteq \mathcal{P}(\mathcal{L})$ denote the set of filters. A {\em weighting function} on $\mathcal{L}$ is a function $w : \mathcal{P}(\mathcal{L}) \rightarrow [0,1]$ satisfying 

\begin{enumerate}

\item $w(\mathcal{L}) = 1$, and 

\item $w(\bigcup_{i\in I} A_i) = \sum_{i \in I} w(A_i)$ for pairwise disjoint $A_i$ ($i \in I$).

\end{enumerate}

\end{definition}

\begin{definition}\label{def-matching}\rm

A {\em matching measure} is a function $\mu : \mathbb{F} \times \mathbb{F} \rightarrow [0,1]$ such that $\mu(\mathcal{F}_1,\mathcal{F}_2) = w(\mathcal{F}_1 \cap \mathcal{F}_2) / w(\mathcal{F}_2)$ holds for some weighting function $w$ on $\mathcal{L}$.

\end{definition}

\begin{example}\label{bsp-lattice}

Take a simple lattice $\mathcal{L}$ with only five elements: $\mathcal{L} = \{ C_1, C_2, C_3, C_4, C_5 \}$. The lattice structure is shown in Figure \ref{fig-lattice}. Then we obtain seven filters for this lattice, each generated by one or two elements of the lattice. These filters are also shown in Figure \ref{fig-lattice}.

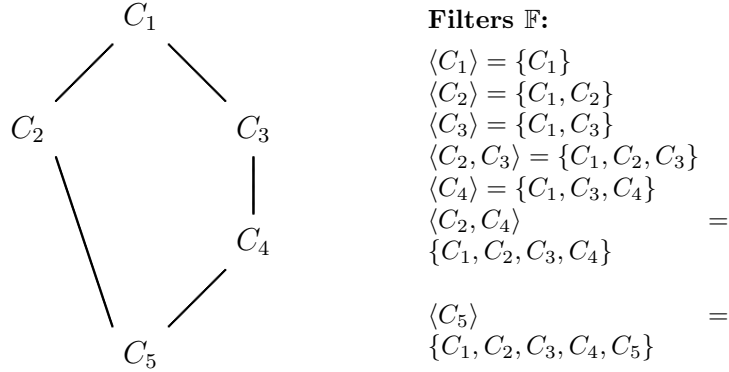
\begin{figure}
\begin{minipage}{6cm}
\begin{center}
\unitlength7.5mm
\begin{picture}(8,6)
\thicklines
\put(4,0){\makebox(0,0){\large $C_5$}}
\put(6,2){\makebox(0,0){\large $C_4$}}
\put(2,4){\makebox(0,0){\large $C_2$}}
\put(6,4){\makebox(0,0){\large $C_3$}}
\put(4,6){\makebox(0,0){\large $C_1$}}
\qbezier(4.5,0.5)(5,1)(5.5,1.5)
\qbezier(3.5,0.5)(3,2)(2.5,3.5)
\qbezier(2.5,4.5)(3,5)(3.5,5.5)
\qbezier(5.5,4.5)(5,5)(4.5,5.5)
\qbezier(6,2.5)(6,3)(6,3.5)
\end{picture}
\end{center}
\end{minipage}
\qquad
\begin{minipage}{4cm}
\textbf{Filters $\mathbb{F}$:}\\[0.8ex]
$\langle C_1 \rangle = \{ C_1 \}$\\
$\langle C_2 \rangle = \{ C_1, C_2 \}$\\
$\langle C_3 \rangle = \{ C_1, C_3 \}$\\
$\langle C_2, C_3 \rangle = \{ C_1, C_2, C_3 \}$\\
$\langle C_4 \rangle = \{ C_1, C_3, C_4 \}$\\
$\langle C_2, C_4 \rangle = \{ C_1, C_2, C_3, C_4 \}$\\

$\langle C_5 \rangle = \{ C_1, C_2, C_3, C_4, C_5 \}$
\end{minipage}
\caption{\label{fig-lattice}A simple lattice and its filters}
\end{figure}

\renewcommand{\arraystretch}{1.2}
\begin{table}[h!]
\begin{center}
\begin{tabular}{|r||*7{c|}}
\hline
& $\langle C_1 \rangle$ & $\langle C_2 \rangle$ & $\langle C_3 \rangle$ & $\langle C_2, C_3 \rangle$ & $\langle C_4 \rangle$ & $\langle C_2, C_4 \rangle$ & $\langle C_5 \rangle$ \\ \hline\hline
$\langle C_1 \rangle$ & 1 & $\frac{1}{4}$ & $\frac{1}{3}$ & $\frac{1}{6}$ & $\frac{1}{6}$ & $\frac{1}{9}$ & $\frac{1}{10}$ \\ \hline
$\langle C_2 \rangle$ & 1 & 1 & $\frac{1}{3}$ & $\frac{2}{3}$ & $\frac{1}{6}$ & $\frac{4}{9}$ & $\frac{2}{5}$ \\ \hline
$\langle C_3 \rangle$ & 1 & $\frac{1}{4}$ & 1 & $\frac{1}{2}$ & $\frac{1}{2}$ & $\frac{1}{3}$ & $\frac{3}{10}$ \\ \hline
$\langle C_2, C_3 \rangle$ & 1 & 1 & 1 & 1 & $\frac{1}{2}$ & $\frac{2}{3}$ & $\frac{3}{5}$ \\ \hline
$\langle C_4 \rangle$ & 1 & $\frac{1}{4}$ & 1 & $\frac{1}{2}$ & 1 & $\frac{2}{3}$ & $\frac{2}{5}$ \\ \hline
$\langle C_2, C_4 \rangle$ & 1 & 1 & 1 & 1 & 1 & 1 & $\frac{9}{10}$ \\ \hline
$\langle C_5 \rangle$ & 1 & 1 & 1 & 1 & 1 & 1 & 1 \\ \hline
\end{tabular}
\caption{\label{tab-matching}A matching measure $\mu$ on the lattice $\mathcal{L}$}
\end{center}
\end{table}

If we now define weights $w(C_1) = \frac{1}{10}$, $w(C_2) = \frac{3}{10}$, $w(C_3) = \frac{1}{5}$, $w(C_4)= \frac{3}{10}$, $w(C_5) = \frac{1}{10}$, then we obtain the matching measure values $\mu(\mathcal{F},\mathcal{G})$ shown in Table \ref{tab-matching}. In the table the row label is $\mathcal{F}$ and the column label is $\mathcal{G}$.

\end{example}

Obviously, every matching measure $\mu$ is defined by weights $w(C) = w(\{ C \}) \in [0,1]$ for the elements $C \in \mathcal{L}$. With this we immediately obtain $w(\mathcal{F}) = \sum_{C \in \mathcal{F}} w(C)$ and $w(\mathcal{L} - \mathcal{F}) = 1- w(\mathcal{F})$. Then $\mu(\mathcal{F}_1,\mathcal{F}_2) =  \sum_{C \in \mathcal{F}_1 \cap \mathcal{F}_2} w(C) \cdot \left ( \sum_{C \in \mathcal{F}_2} w(C) \right )^{-1}$ expresses, how much of (requested) profile represented by $\mathcal{F}_2$ is contained in the (given) profile represented by $\mathcal{F}_1$.

\begin{example} \label{bsp-popov}

The matching measure $\mu_{pj}$ defined in \cite{popov:2013} uses simply cardinalities:
\[ \mu_{pj}(\mathcal{F}_1,\mathcal{F}_2) = \#(\mathcal{F}_1 \cap \mathcal{F}_2) / \#\mathcal{F}_2 \]

Thus, it is defined by the weighting function $w$ on $\mathcal{L}$ with $w(A) = \# A / \# \mathcal{L}$, i.e. all properties have equal weights.

\end{example}

Note that in general matching measures are not symmetric. If $\mu(\mathcal{F}_1,\mathcal{F}_2)$ expresses how well a given profile (represented by $\mathcal{F}_1$) matches a requested profile (represented by $\mathcal{F}_2$), then $\mu(\mathcal{F}_2,\mathcal{F}_1)$ measures what is ``too much'' in the given profile that is not required in the requested profile.

\begin{example}\label{bsp-overq}

Take the lattice $\mathcal{L}$ from Example \ref{bsp-lattice} shown in Figure \ref{fig-lattice}. Let $\mathcal{G} = \langle C_2, C_3 \rangle$ be the filter that represents the requirements. If we take filters $\mathcal{F}_1 = \langle C_3 \rangle$ and $\mathcal{F}_2 = \langle C_4 \rangle$ representing given profiles, then we have $\mu(\mathcal{F}_1,\mathcal{G}) = \mu(\mathcal{F}_2,\mathcal{G}) = \frac{1}{2}$, i.e. both given filters match the requirements equally well. However, if we also consider the {\em inverse matching measure} $\bar{\mu}$ with values $\bar{\mu}(\mathcal{F}_1,\mathcal{G}) = \mu(\mathcal{G},\mathcal{F}_1) = 1$ and $\bar{\mu}(\mathcal{F}_2,\mathcal{G}) = \mu(\mathcal{G},\mathcal{F}_2) = \frac{1}{2}$, then we see that $\mathcal{F}_1$ matches ``better'', as $\mathcal{F}_2$ contains $C_4$, which is not required at all.

\end{example}

The example shows that it may make sense to consider not just a single matching measure $\mu$, but also its inverse---In the recruiting area this corresponds to ``over-qualification'' \cite{paoletti:dexa2015}---$\bar{\mu}$, several matching measures defined on different lattices, or weight\-ed aggregates of such measures.

\subsection{The Lattice of Matching Value Terms}

We know that a matching measure $\mu$ on $\mathbb{F}$ is defined by weights $w(C)$ for each concept $C \in \mathcal{L}$, i.e.
\[ \mu(\mathcal{F},\mathcal{G}) = \cfrac{\sum_{C \in \mathcal{F} \cap \mathcal{G}} w(C)}{\sum_{C \in \mathcal{G}} w(C)} \]

with $w(C) > 0$. In case $\mathcal{G} \subseteq \mathcal{F}$, the right hand side becomes 1. In all other cases the actual value on the right hand side depends on the weights $w(C)$. Let us call the right hand side expression (including 1) a {\em matching value term} (mvt).

\begin{example}\label{bsp-mvt}

Let us look at the lattice $\mathcal{L}$ from Example \ref{bsp-lattice}. With $w(C_1) = a$, $w(C_2) = b$, $w(C_3) = c$, $w(C_4)= d$ and $w(C_5) = e$ we obtain the matching value terms shown in Table \ref{tab-mvt}.

\renewcommand{\arraystretch}{1.4}
\begin{table}
\begin{center}
\begin{tabular}{|r||*7{c|}}
\hline
& $\langle C_1 \rangle$ & $\langle C_2 \rangle$ & $\langle C_3 \rangle$ & $\langle C_2, C_3 \rangle$ & $\langle C_4 \rangle$ & $\langle C_2, C_4 \rangle$ & $\langle C_5 \rangle$ \\ \hline\hline
$\langle C_1 \rangle$ & 1 & $\frac{a}{a+b}$ & $\frac{a}{a+c}$ & $\frac{a}{a+b+c}$ & $\frac{a}{a+c+d}$ & $\frac{a}{a+b+c+d}$ & $\frac{a}{a+b+c+d+e}$ \\ \hline
$\langle C_2 \rangle$ & 1 & 1 & $\frac{a}{a+c}$ & $\frac{a+b}{a+b+c}$ & $\frac{a}{a+c+d}$ & $\frac{a+b}{a+b+c+d}$ & $\frac{a+b}{a+b+c+d+e}$ \\ \hline
$\langle C_3 \rangle$ & 1 & $\frac{a}{a+b}$ & 1 & $\frac{a+c}{a+b+c}$ & $\frac{a+c}{a+c+d}$ & $\frac{a+c}{a+b+c+d}$ & $\frac{a+c}{a+b+c+d+e}$ \\ \hline
$\langle C_2, C_3 \rangle$ & 1 & 1 & 1 & 1 & $\frac{a+c}{a+c+d}$ & $\frac{a+b+c}{a+b+c+d}$ & $\frac{a+b+c}{a+b+c+d+e}$ \\ \hline
$\langle C_4 \rangle$ & 1 & $\frac{a}{a+b}$ & 1 & $\frac{a+c}{a+b+c}$ & 1 & $\frac{a+c+d}{a+b+c+d}$ & $\frac{a+c+d}{a+b+c+d+e}$ \\ \hline
$\langle C_2, C_4 \rangle$ & 1 & 1 & 1 & 1 & 1 & 1 & $\frac{a+b+c+d}{a+b+c+d+e}$ \\ \hline
$\langle C_5 \rangle$ & 1 & 1 & 1 & 1 & 1 & 1 & 1 \\ \hline
\end{tabular}
\caption{\label{tab-mvt}Matching measure terms for the lattice $\mathcal{L}$}
\end{center}
\end{table}

\end{example} 

\begin{definition}

Let $\mathbb{V}$ denote the set of all matching value terms (including 1) for the lattice $\mathcal{L}$. A partial order $\le$ on $\mathbb{V}$ is defined by $v_1 \le v_2$ iff each substitution of positive values for $w(C)$ results in values $\bar{v}_1, \bar{v}_2 \in [0,1]$ with $\bar{v}_1 \le \bar{v}_2$.

\end{definition}

Let us first look at the following three special cases:

\begin{enumerate}

\item If $v_2$ differs from $v_1$ by a summand $w(C)$ added to the nominator, i.e. $\mathcal{F}_2 = \mathcal{F}_1 \cup \{C\}$ and $\mathcal{G}_2 = \mathcal{G}_1$, then we obviously obtain $v_1 \le v_2$.

\item If $v_1$ differs from $v_2$ by a summand $w(C)$ added to the denominator, i.e. $\mathcal{G}_1 = \mathcal{G}_2 \cup \{C\}$ and $\mathcal{F}_2 = \mathcal{F}_1$, then we obviously obtain $v_1 \le v_2$.

\item If $v_2$ differs from $v_1$ by a summand $w(C)$ added to both the nominator and the denominator, i.e. $\mathcal{G}_2 = \mathcal{G}_1 \cup \{C\}$ and $\mathcal{F}_2 = \mathcal{F}_1 \cup \{C\}$, then we also obtain $v_1 \le v_2$.

\end{enumerate}

\begin{lemma}\label{lem-v1}

The partial order $\le$ on the set $\mathbb{V}$ of matching value terms is the reflexive, transitive closure of the relation defined by the cases (1), (2) and (3) above.

\end{lemma}

\begin{proof}

Let $v_i = \cfrac{\sum_{C \in \mathcal{F}_i} w(C)}{\sum_{C \in \mathcal{G}_i} w(C)}$ with $\mathcal{F}_i \subseteq \mathcal{G}_i$ ($i=1,2$). If $v_2$ results from $v_1$ by a sequence of the operations (1), (2) and (3) above, we obviously get $\mathcal{F}_2 = \mathcal{F}_1 \cup \mathcal{H}_1 \cup \mathcal{H}_3$ and $\mathcal{G}_2 = (\mathcal{G}_1 \cup \mathcal{H}_3) \mathcal{H}_2$ subject to the conditions $\mathcal{H}_1 \subseteq \mathcal{G}_1 - \mathcal{H}_2$, $ \mathcal{H}_2 \subseteq \mathcal{G}_1 = \mathcal{F}_1$ and $\mathcal{H}_3 \cap \mathcal{G}_1 = \emptyset$. Thus, if $v_2$ does not result from $v_1$ by such a sequence of operations, we either find a $C \in \mathcal{G}_2 - \mathcal{F}_2$ with $C \notin \mathcal{G}_1$ or there is some $D \in \mathcal{F}_1$ with $D \notin \mathcal{F}_2 \cap \mathcal{G}_1$. In the first case $\bar{v}_2$ can be made arbitrarily small by a suitable substitution, in the second case $\bar{v}_1$ can be made arbitrarily large. Therefore, we must have $v_1 \not\le v_2$.

\end{proof}

With this characterisation of the partial order on $\mathbb{V}$ we can show that $\mathbb{V}$ is in fact a lattice.

\begin{theorem}\label{lem-v2}

$(\mathbb{V},\le)$ is a lattice with top element 1 and bottom element $\cfrac{w(\top)}{\sum_{C \in \mathcal{L}} w(C)}$. The join is given by $v_1 \sqcup v_2 = \cfrac{\sum_{C \in \mathcal{F}_1 \cup \mathcal{F}_2} w(C)}{\sum_{C \in (\mathcal{G}_1 \cup \mathcal{F}_2) \cap (\mathcal{F}_1 \cup \mathcal{G}_2)} w(C)}$, and the meet is given by $v_1 \sqcap v_2 = \cfrac{\sum_{C \in \mathcal{F}_1 \cap \mathcal{F}_2} w(C)}{\sum_{C \in \langle(\mathcal{G}_1 - \mathcal{F}_1) \cup (\mathcal{G}_2 - \mathcal{F}_2) \cup (\mathcal{F}_1 \cap \mathcal{F}_2)\rangle} w(C)}$, where $\langle S \rangle$ denotes the filter generated by the subset $S$.

\end{theorem}

\begin{proof}

The expressions for meet and join are well-defined, as the sum in the nominator always ranges over a subset of the range of the sum in the denominator.

Now let $v_i = \cfrac{\sum_{C \in \mathcal{F}_i} w(C)}{\sum_{C \in \mathcal{G}_i} w(C)}$ with $\mathcal{F}_i \subseteq \mathcal{G}_i$. For the join the summands $w(C)$ in the nominator of $v_1 \sqcup v_2$ not appearing in the nominator $v_i$ are those $C \in \mathcal{F}_j - \mathcal{F}_i$ for $\{i,j\} = \{1,2\}$. Among these, $C \in \mathcal{F}_j - \mathcal{G}_i$ define summands also added to the denominator of $v_1 \sqcup v_2$, i.e. each such $C$ defines an operation of type (3) above to increase $v_i$. The remaining $C \in (\mathcal{F}_j \cap \mathcal{G}_i) - \mathcal{F}_i$ define operations of type (1) above to increase $v_i$. Furthermore, $C \in \mathcal{G}_i - \mathcal{G}_j$ define those summands $w(C)$ in the denominator of $v_i$ that do not appear in the denominator of $v_1 \sqcup v_2$, so they define operations of type (2) above to increase $v_i$. That is, $v_1 \sqcup v_2$ results from $v_i$ by a sequence of operations of types (1), (2) and (3), which shows $v_i \le v_1 \sqcup v_2$.

Conversely, let $v_1, v_2 \le v^\prime$. As $v^\prime$ must result from $v_i$ by a sequence of operations of types (1), (2) and (3), it must contain all summands $w(C)$ with $C \in \mathcal{F}_1 \cup \mathcal{F}_2$ in its nominator, and these summands must then also appear in its denominator. Furthermore, as a consequence of increasing $v_i$ by operations of type (2), the denominator must not contain summands $w(C)$ with $C \in (\mathcal{G}_1 \cup \mathcal{G}_2) - (\mathcal{G}_1 \cap \mathcal{G}_2)$ unless $C \in \mathcal{F}_1 \cup \mathcal{F}_2$. That is, the denominator of $v^\prime$ only contains summands $w(C)$ with $C \in (\mathcal{G}_1 \cap \mathcal{G}_2) \cup \mathcal{F}_1 \cup \mathcal{F}_2$, which define the denominator of $v_1 \sqcup v_2$. so we obtain $v_1 \sqcup v_2 \le v^\prime$, which proves that the join is indeed defined by the equation above.

For the meet we can argue analogously. Comparing $v_i$ with $v_1 \sqcap v_2$ the former one contains additional summands $w(C)$ in its nominator with $C \in \mathcal{F}_i - \mathcal{F}_j$. Out of these, if $C \in \mathcal{F}_i \cap (\mathcal{G}_j - \mathcal{F}_j)$, then summands only appear in the nominator, which gives rise to an operation of type (1) increasing $v_1 \sqcap v_2$. If $C \in \mathcal{F}_i - \mathcal{G}_j$ holds, then the summand $w(C)$ has been added to both the nominator and denominator of $v_i$, which gives rise to an operation of type (3) to increase $v_1 \sqcap v_2$. If $C \in \mathcal{G}_j - \mathcal{G}_i - \mathcal{F}_j$ holds, then the summand $w(C)$ appears in the denominator of $v_1 \sqcap v_2$, but not in the denominator of $v_i$, which defines an operation of type (2) to increase $v_1 \sqcap v_2$. That is, $v_i$ results from $v_1 \sqcap v_2$ by a sequence of operations of type (1), (2) and (3) above, which shows $v_1 \sqcap v_2 \le v_i$.

Conversely, if $v^\prime \le v_1, v_2$ holds, then the nominator of $v^\prime$ can only contain summands $w(C)$ with $C \in \mathcal{F}_1 \cap \mathcal{F}_2$. If we remove all $C \in \mathcal{F}_i - \mathcal{F}_j$ also from the denominator $\mathcal{G}_i$, which corresponds to operation (3), we obtain $ \mathcal{G}_i - (\mathcal{F}_i - \mathcal{F}_j) = (\mathcal{G}_i - \mathcal{F}_i) \cup (\mathcal{G}_i \cap \mathcal{F}_j)$. Furthermore, operations of type (2) that decrease $v_i$ can only add summands to the denominator, so in total we get the union $(\mathcal{G}_1 - \mathcal{F}_1) \cup (\mathcal{G}_2 - \mathcal{F}_2) \cup (\mathcal{F}_1 \cap \mathcal{F}_2)$. However, operations of type (3) can only be applied, if nominator and denominator are defined by filters. This shows $v^\prime \le v_1 \sqcap v_2$, which proves that the meet is indeed defined by the equation above.

The statements concerning the top and bottom element in $\mathbb{V}$ with respect to $\le$ are obvious.

\end{proof}

\begin{figure}
\unitlength1.1cm
\begin{center}
\begin{picture}(10.5,9)
\thicklines
\put(5,9){\makebox(0,0){$1$}}
\put(2,7.5){\makebox(0,0){$\frac{a+b}{a+b+c}$}}
\put(4,7.5){\makebox(0,0){$\frac{a+c+d}{a+b+c+d}$}}
\put(7,7.5){\makebox(0,0){$\frac{a+b+c}{a+b+c+d}$}}
\put(9,7.5){\makebox(0,0){$\frac{a+b+c+d}{a+b+c+d+e}$}}
\put(1,6){\makebox(0,0){$\frac{a}{a+c}$}}
\put(2.5,6){\makebox(0,0){$\frac{a+b}{a+b+c+d}$}}
\put(4,6){\makebox(0,0){$\frac{a+c}{a+b+c}$}}
\put(5.5,6){\makebox(0,0){$\frac{a+c}{a+c+d}$}}
\put(7,6){\makebox(0,0){$\frac{a+b+c}{a+b+c+d+e}$}}
\put(9,6){\makebox(0,0){$\frac{a+c+d}{a+b+c+d+e}$}}
\put(2,4.5){\makebox(0,0){$\frac{a}{a+b}$}}
\put(4,4.5){\makebox(0,0){$\frac{a}{a+c+d}$}}
\put(5.25,4.5){\makebox(0,0){$\frac{a+c}{a+b+c+d}$}}
\put(1,3){\makebox(0,0){$\frac{a}{a+b+c}$}}
\put(7,3){\makebox(0,0){$\frac{a+c}{a+b+c+d+e}$}}
\put(2,1.5){\makebox(0,0){$\frac{a}{a+b+c+d}$}}
\put(4,1.5){\makebox(0,0){$\frac{a+b}{a+b+c+d+e}$}}
\put(4,0){\makebox(0,0){$\frac{a}{a+b+c+d+e}$}}
\qbezier(2.3,7.65)(3,8)(4.8,8.9)
\qbezier(4.2,7.7)(4.5,8,25)(4.9,8.9)
\qbezier(6.8,7.7)(6,8.25)(5.1,8.9)
\qbezier(8.7,7.65)(7,8.25)(5.2,8.9)
\qbezier(1.1,6.15)(1.5,6.75)(1.9,7.35)
\qbezier(2.45,6.15)(2.25,6.75)(2.05,7.35)
\qbezier(2.95,6.15)(4,6.5)(6.55,7.35)
\qbezier(4,6.2)(4,6.75)(4,7.3)
\qbezier(5.7,6.2)(6,6.5)(6.8,7.3)
\qbezier(7,6.2)(7,6.5)(7,7.3)
\qbezier(7.14,6.21)(8,6.75)(8.86,7.29)
\qbezier(8.5,6.15)(6.5,6.75)(4.5,7.35)
\qbezier(9,6.2)(9,6.75)(9,7.3)
\qbezier(2.2,4.65)(3,5.25)(3.8,5.85)
\qbezier(3.7,4.65)(3,5)(1.3,5.85)
\qbezier(3.8,4.7)(3.5,5)(2.7,5.8)
\qbezier(4.2,4.7)(4.5,5)(5.3,5.8)
\qbezier(5.075,4.65)(4.675,5.25)(4.175,5.85)
\qbezier(5.275,4.65)(5.375,5.25)(5.475,5.85)
\qbezier(1,3.2)(1,4.5)(1,5.8)
\qbezier(1.1,3.15)(1.5,3.75)(1.9,4.35)
\qbezier(6.825,3.15)(6.125,3.75)(5.425,4.35)
\qbezier(7,3.2)(7,4.5)(7,5.8)
\qbezier(7.1,3.15)(8,4.5)(8.9,5.85)
\qbezier(1.9,1.65)(1.5,2.25)(1.1,2.85)
\qbezier(2.1,1.65)(3,3)(3.9,4.35)
\qbezier(2.35,1.65)(3.7,3)(5.05,4.35)
\qbezier(3.9,1.8)(3.5,3)(2.6,5.7)
\qbezier(4.2,1.8)(5,3)(6.8,5.7)
\qbezier(3.8,0.15)(3,0.75)(2.2,1.35)
\qbezier(4,0.2)(4,0.75)(4,1.3)
\qbezier(4.2,0.2)(5,1)(6.8,2.8)
\end{picture}
\caption{\label{fig-mvt}The lattice of matching value terms for the lattice $\mathcal{L}$}
\end{center}
\end{figure}
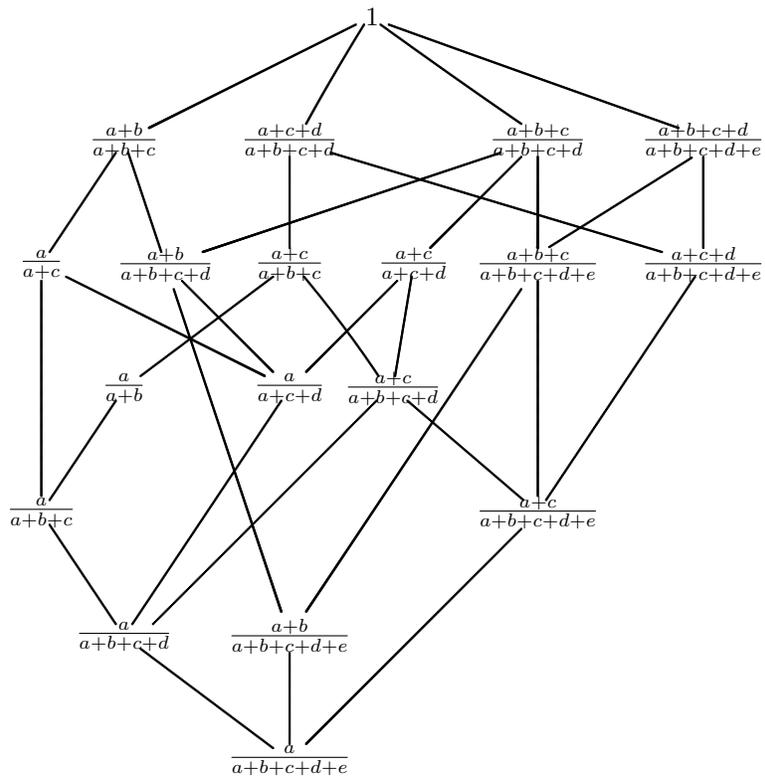

\begin{example}

\ Figure \ref{fig-mvt} shows the lattice $(\mathbb{V},\le)$ of matching value terms for the lattice $\mathcal{L}$ and the set of filters $\mathbb{F}$ from Example \ref{bsp-lattice}.

\end{example}

\begin{definition}\label{def-admissible}\rm

A relation $r \subseteq \mathbb{F} \times \mathbb{F}$ is called {\em admissible} iff the following conditions hold:

\begin{enumerate}

\item If $\mathcal{G} \subseteq \mathcal{F}$ holds, then $r(\mathcal{F},\mathcal{G})$ holds.

\item If $r(\mathcal{F},\mathcal{G})$ holds and $C \notin \mathcal{F}$, then also $r(\mathcal{F},\mathcal{G} - \{C\})$ holds.

\item If $r(\mathcal{F},\mathcal{G})$ holds, then also $r(\mathcal{F} \cup \{C\},\mathcal{G} \cup \{C\})$ holds for all $C \in \mathcal{L}$.

\end{enumerate}

\end{definition}

\noindent
Let us concentrate on a single admissible relation $r \subseteq \mathbb{F} \times \mathbb{F}$.

\begin{lemma}\label{lem-mvt}

Let $r \subseteq \mathbb{F} \times \mathbb{F}$ be an admissible relation. Define the set
\[ \mathcal{R} = \{ v \in \mathbb{V} \mid \exists \mathcal{F}, \mathcal{G} \in \mathbb{F} .\; v = mvt(\mathcal{F},\mathcal{G}) \wedge r(\mathcal{F},\mathcal{G}) \} \]

of matching value terms. Then $\mathcal{R}$ is a filter in the lattice $(\mathbb{V},\le)$.

\end{lemma}

\begin{proof}

\ Let $v_1 \in \mathcal{R}$, say $v_1 = mvt(\mathcal{F}_1,\mathcal{G}_1)$ such that $r(\mathcal{F}_1,\mathcal{G}_1)$ holds. Let $v_1 \le v_2$ for some other mvt $v_2 \in \mathbb{V}$. Without loss of generality we can assume that $v_1$ and $v_2$ differ by one of the three possible cases (1), (2) and (3) used for Lemma \ref{lem-v1}. We have to show $v_2 \in \mathcal{R}$.

\begin{enumerate}

\item In this case $v_2$ differs from $v_1$ by a summand $w(C)$ added to the nominator, i.e. $v_2 = mvt(\mathcal{F}_2,\mathcal{G}_1)$ with $\mathcal{F}_2 = \mathcal{F}_1 \cup \{ C \}$ for $C \in \mathcal{G}_1 - \mathcal{F}_1$. Then $r(\mathcal{F}_2,\mathcal{G}_1)$ holds due to property (3) of Definition \ref{def-admissible}, which gives $v_2 \in \mathcal{R}$ in this case.

\item In this case $v_1$ differs from $v_2$ by a summand $w(C)$ added to the denominator, i.e. $v_2 = mvt(\mathcal{F}_1,\mathcal{G}_2)$ with $\mathcal{G}_2 = \mathcal{G}_1 - \{ C \}$ for some $C \notin \mathcal{F}_1$. Then $r(\mathcal{F}_1,\mathcal{G}_2)$ holds due to property (2) of Definition \ref{def-admissible}, which gives $v_2 \in \mathcal{R}$ also in this case.

\item In this case $v_2$ differs from $v_1$ by a summand $w(C)$ added to both the nominator and the denominator, i.e. $v_2 = mvt(\mathcal{F}_2,\mathcal{G}_2)$ with $\mathcal{G}_2 = \mathcal{G}_1 \cup \{ C \}$ and $\mathcal{F}_2 = \mathcal{F}_1 \cup \{ C \}$ for some $C \notin \mathcal{G}_1$. Again, due to property (3) of Definition \ref{def-admissible} we obtain $r(\mathcal{F}_2,\mathcal{G}_2)$, which gives $v_2 \in \mathcal{R}$ and completes the proof.

\end{enumerate}

\end{proof}

\section{Learning Matching Measures from User-Defined Matchings}\label{sec:learning}

In the previous Section \ref{sec:matching} we developed a general theory of matching exploiting filters in lattices. Such lattices can be derived from knowledge bases as shown in Section \ref{sec:dl}. We now address the problem how to learn a matching measure from matching values that are given by a human domain expert. For this let $(\mathcal{L},\le)$ be a finite lattice, and let $\mathbb{F}$ denote the set of all its filters. Note that each filter $\mathcal{F} \in \mathbb{F}$ is uniquely determined by its minimal elements, so we can write $\mathcal{F} = \langle C_1 ,\dots, C_k \rangle$. The matching knowledge of a human expert can be represented be a mapping $h: \mathbb{F} \times \mathbb{F} \rightarrow [0,1]$. Though human experts will hardly ever provide complete information, we will assume in the sequel that $h$ is total. 

However, the matching values as such are merely used to determine rankings, whereas their concrete value is of minor importance. Therefore, instead of asking whether there exists a matching measure $\mu$ on $\mathbb{F}$ such that $\mu(\mathcal{F},\mathcal{G}) = h(\mathcal{F},\mathcal{G})$ holds for all pairs of filters, we investigate the slightly weaker problem to find a {\em ranking-preserving} matching measure $\mu$ on $\mathbb{F}$, i.e. the matching measure should imply the same rankings. 

At least this is the case for rankings with respect to a fixed requested profile. For a fixed given profile a bit more care is needed, as the following example shows.

\begin{example}

Consider the following example from the recruiting domain. Let
\begin{align*}
\mathcal{F} &= \{ \text{Skill}, \text{Roman\_language}, \text{Programming} \} ,\\
\mathcal{G}_1 &= \{ \text{Skill}, \text{Roman\_language}, \text{Italian} \} \qquad \text{and} \\
\mathcal{G}_2 &= \{ \text{Skill}, \text{Programming}, \text{Java} \} .
\end{align*}

It makes perfectly sense to rank {\em given} profiles $\mathcal{G}_1$ and $\mathcal{G}_2$ with respect to a {\em requested} profile $\mathcal{F}$, i.e. to consider $h(\mathcal{G}_1,\mathcal{F})$ and $h(\mathcal{G}_2,\mathcal{F})$ and to preserve a ranking between these two in a weighted matching measure, even if such a requested profile appears to be rather odd. 

However, if $\mathcal{F}$ is a given profile (which may not be unrealistic), then it appears doubtful, if a ranking for $\mathcal{G}_1$ (emphasising language skills) and $\mathcal{G}_2$ (emphasising programming skills) makes sense at all.

\end{example}

Therefore, it seems plausible to classify profiles with respect to their {\em relevance} for $\mathcal{F}$ using an equivalence relation $\sim_{\mathcal{F}}$ on $\mathbb{F}$ defined as $\mathcal{G}_1 \sim_{\mathcal{F}} \mathcal{G}_2$ iff $\mathcal{F} \cap \mathcal{G}_1 = \mathcal{F} \cap \mathcal{G}_2$. In doing so our claim becomes that rankings by $h$ should be preserved for a given profile $\mathcal{F}$ in {\em relevance classes} for $\mathcal{F}$.

\begin{definition}\label{def-ranking}\rm

A matching measure $\mu$ on $\mathbb{F}$ is called {\em ranking-preserving} with respect to $h: \mathbb{F} \times \mathbb{F} \rightarrow [0,1]$ if for all filters the following two conditions hold:

\begin{enumerate}

\item $\mu(\mathcal{F}_1,\mathcal{G}) > \mu(\mathcal{F}_2,\mathcal{G})$ holds, whenever $h(\mathcal{F}_1,\mathcal{G}) > h(\mathcal{F}_2,\mathcal{G})$ holds;

\item $\mu(\mathcal{F}, \mathcal{G}_1) > \mu(\mathcal{F}, \mathcal{G}_2)$ holds, whenever $h(\mathcal{F}, \mathcal{G}_1) > h(\mathcal{F}, \mathcal{G}_2)$ holds, provided that $\mathcal{G}_1$ and $\mathcal{G}_2$ are in the same relevance class with respect to $\mathcal{F}$, i.e. $\mathcal{G}_1 \sim_{\mathcal{F}} \mathcal{G}_2$.

\end{enumerate}

\end{definition}

Thus, this section is dedicated to finding conditions for $h$ that guarantee the existence of a ranking-preserving matching measure $\mu$ with respect to $h$.

\subsection{Plausibility Constraints}

We are looking for plausibility constraints for the mapping $h$ that should be satisfied in the absence of bias, i.e. the assessment of the human expert is not grounded in hidden concepts. If such plausibility conditions are satisfied we explore the existence of a ranking-preserving matching measure $\mu$. First we show the following simple lemma.

\begin{lemma}\label{lem-plausible}

Let $\mu$ be a matching measure on $\mathbb{F}$. Then for all filters $\mathcal{F}, \mathcal{F}_1$, $\mathcal{F}_2, \mathcal{G} \in \mathbb{F}$ the following conditions hold, provided that the arguments of $\mu$ are filters:

\begin{enumerate}

\item $\mu(\mathcal{F},\mathcal{G}) =1$ for $\mathcal{G} \subseteq \mathcal{F}$.

\item $\mu(\mathcal{F},\mathcal{G}) = \mu(\mathcal{F} \cap \mathcal{G}, \mathcal{G})$.

\item $\mu(\mathcal{F},\mathcal{G}) < \mu(\mathcal{F},\mathcal{G} - \{C\})$ holds for $C \in \mathcal{G}\setminus \mathcal{F}$.

\item $\mu(\mathcal{F},\mathcal{G}) \le \mu(\mathcal{F} \cup \{C\},\mathcal{G} \cup \{C\})$.

\item If $\mu(\mathcal{F}_1, \mathcal{G}) < \mu(\mathcal{F}_2, \mathcal{G})$ holds, then $\mu(\mathcal{F}_1 \cup \{ C \}, \mathcal{G}) < \mu(\mathcal{F}_2 \cup \{ C \}, \mathcal{G})$ holds for every $C \in \mathcal{G} - \mathcal{F}_1 - \mathcal{F}_2$.

\item If $\mathcal{F} \cap \mathcal{F}_1 \cap \mathcal{G} = \mathcal{F} \cap \mathcal{F}_2 \cap \mathcal{G}$ holds, then $\mu(\mathcal{F}_1, \mathcal{G}) > \mu(\mathcal{F}_2, \mathcal{G}) \Leftrightarrow \mu(\mathcal{F}, \mathcal{F}_1 \cap \mathcal{G}) < \mu(\mathcal{F}, \mathcal{F}_2 \cap \mathcal{G})$.

\item If $\mu(\mathcal{F}, \mathcal{G}_1) < \mu(\mathcal{F}, \mathcal{G}_2)$ holds, then for every $C \in \mathcal{G}_1 \cap \mathcal{G}_2$ also $\mu(\mathcal{F} \cup \{ C \}, \mathcal{G}_1) < \mu(\mathcal{F} \cup \{ C \}, \mathcal{G}_2)$ holds, provided that $\mathcal{F} \cap \mathcal{G}_1 = \mathcal{F} \cap \mathcal{G}_2$ holds.

\item If $\mu(\mathcal{F}_1,\mathcal{G}) < \mu(\mathcal{F}_2,\mathcal{G})$ and $\mathcal{G}^\prime \cap \mathcal{F}_i = \mathcal{G} \cap \mathcal{F}_i$ hold, then also $\mu(\mathcal{F}_1,\mathcal{G}^\prime) < \mu(\mathcal{F}_2,\mathcal{G}^\prime)$ holds.

\end{enumerate}

\end{lemma}

\begin{proof}

Properties (1), (2), (5) and (8) are obvious from the Definition \ref{def-matching} of matching measures. For property (6) both sides of the equivalence are equivalent to $w(\mathcal{F}_1 \cap \mathcal{G}) > w(\mathcal{F}_2 \cap \mathcal{G})$.

For property (3)  we have
\[ \mu(\mathcal{F},\mathcal{G}) = \cfrac{w(\mathcal{F} \cap \mathcal{G})}{w(\mathcal{G} - \{C\}) + w(C)} < \cfrac{w(\mathcal{F} \cap (\mathcal{G} - \{C\}))}{w(\mathcal{G} - \{C\})} = \mu(\mathcal{F},\mathcal{G} - \{C\}) \; . \]

For property (4) the case $C \in \mathcal{F}$ is trivial. In case $C \in \mathcal{G} - \mathcal{F}$ holds, we get
\[ \mu(\mathcal{F},\mathcal{G}) = \cfrac{w(\mathcal{F} \cap \mathcal{G})}{w(\mathcal{G})} \le \cfrac{w(\mathcal{F} \cap \mathcal{G}) + w(C)}{w(\mathcal{G})} = \mu(\mathcal{F} \cup \{C\},\mathcal{G} \cup \{C\}) \; . \]

In case $C \notin \mathcal{G}$ first note that for any values $a,b,c$ with $a \le b$ we get $ab + ac \le ab + bc$ and thus $\cfrac{a}{b} \le \cfrac{a+c}{b+c}$. 

Thus, we get 
$\mu(\mathcal{F},\mathcal{G}) = \cfrac{w(\mathcal{F} \cap \mathcal{G})}{w(\mathcal{G})} \le \cfrac{w(\mathcal{F} \cap \mathcal{G}) + w(C)}{w(\mathcal{G}) + w(C)} = \mu(\mathcal{F} \cup \{C\},\mathcal{G} \cup \{C\})$.

For property (7) assume that we have $\mu(\mathcal{F},\mathcal{G}_1) < \mu(\mathcal{F},\mathcal{G}_2)$. This is equivalent with $w(\mathcal{G}_1)>w(\mathcal{G}_2)$ by the definition of $\mu$ and $\mathcal{F} \cap \mathcal{G}_1 = \mathcal{F} \cap \mathcal{G}_2$. On the other hand, for $C \in \mathcal{G}_1 \cap \mathcal{G}_2$, $(\mathcal{F}\cup \{ C \}) \cap \mathcal{G}_1 = (\mathcal{F}\cup \{ C \}) \cap \mathcal{G}_2$ also holds.

That is, we obtain $\mu(\mathcal{F} \cup \{ C \}, \mathcal{G}_1) < \mu(\mathcal{F} \cup \{ C \}, \mathcal{G}_2)$ as claimed.

\end{proof}

Informally phrased property (1) states that whenever all requirements in a requested profile $\mathcal{G}$ (maybe even more) are satisfied by a given profile $\mathcal{F}$, then $\mathcal{F}$ is a perfect match for $\mathcal{G}$. Property (2) states that the matching value indicating how well the given profile $\mathcal{F}$ fits to the requested one $\mathcal{G}$ only depends on $\mathcal{F} \cap \mathcal{G}$, i.e. the properties in the given profile that are relevant for the requested one. Property (3) states that if a requirement not satisfied by a given profile $\mathcal{F}$ is removed from the requested profile $\mathcal{G}$, the given profile will become a better match for the restricted profile. Property (4) covers two cases. If $C \in \mathcal{G}$ holds, then simply the profile $\mathcal{F} \cup \{C\}$ satisfies more requirements than $\mathcal{F}$, so the matching value should increase. The case $C \notin \mathcal{G}$ is a bit more tricky, as the profile $\mathcal{G} \cup \{C\}$ contains an additional requirement, which is satisfied by the enlarged profile $\mathcal{F} \cup \{C\}$. In this case the matching value should increase, because the percentage of requirements that are satisfied increases. Property (5) states that the given profile $\mathcal{F}_1$ is better suited for the required profile $\mathcal{G}$ than the given profile $\mathcal{F}_2$, then adding a new required property $C$ to both given profiles preserves the inequality between the two matching values. Property (6) states that if the given profile $\mathcal{F}_1$ is better suited for the required profile $\mathcal{G}$ than the given profile $\mathcal{F}_2$, then relative to $\mathcal{G}$ the profile $\mathcal{F}_2$ is less over-qualified than $\mathcal{F}_1$ for any other required profile $\mathcal{F}$, provided the intersections of $\mathcal{F}\cap \mathcal{G}$ with the two given profiles coincide. Property (7) states that if $\mathcal{F}$ fits better to $\mathcal{G}_2$ than to $\mathcal{G}_1$ and the relevance of $\mathcal{F}$ with respect to $\mathcal{G}_1$ and $\mathcal{G}_2$ (expressed by the intersection) is the same, then adding property $C$ to $\mathcal{F}$ that is requested in both $\mathcal{G}_1$ and $\mathcal{G}_2$ preserves this dependency, i.e. $\mathcal{F} \cup \{ C \}$ fits better to $\mathcal{G}_2$ than to $\mathcal{G}_1$. Property (8) states that if $\mathcal{F}_2$ fits better to $\mathcal{G}$ than $\mathcal{F}_1$, then this is also the case for any other requested filter $\mathcal{G}^\prime$ that preserves the relevance of $\mathcal{G}$ for both filters $\mathcal{F}_1$ and $\mathcal{F}_2$.

Thus, disregarding for the moment our theory of matching measures, all eight properties in Lemma \ref{lem-plausible} appear to be reasonable. Therefore, we require them as {\em plausibility constraints} that a human-defined mapping $h: \mathbb{F} \times \mathbb{F} \rightarrow [0,1]$ should satisfy.

\begin{definition}\label{def-plausible}

A function $h: \mathbb{F} \times \mathbb{F} \rightarrow [0,1]$ {\em satisfies the plausibility constraints}, if the following conditions are satisfied, provided that the arguments of $h$ are filters:

\begin{enumerate}

\item $h(\mathcal{F},\mathcal{G}) =1$ for $\mathcal{G} \subseteq \mathcal{F}$, 

\item $h(\mathcal{F},\mathcal{G}) = h(\mathcal{F} \cap \mathcal{G}, \mathcal{G})$,

\item $h(\mathcal{F},\mathcal{G}) < h(\mathcal{F},\mathcal{G} - \{C\})$ for any concept $C \in \mathcal{G}\setminus\mathcal{F}$, and 

\item $h(\mathcal{F},\mathcal{G}) \le h(\mathcal{F} \cup \{C\},\mathcal{G} \cup \{C\})$ for any concept $C$.

\item If $h(\mathcal{F}_1, \mathcal{G}) < h(\mathcal{F}_2, \mathcal{G})$ holds, then $h(\mathcal{F}_1 \cup \{ C \}, \mathcal{G}) < h(\mathcal{F}_2 \cup \{ C \}, \mathcal{G})$ holds for every $C \in \mathcal{G} - \mathcal{F}_1 - \mathcal{F}_2$.

\item If $\mathcal{F} \cap \mathcal{F}_1 \cap \mathcal{G} = \mathcal{F} \cap \mathcal{F}_2 \cap \mathcal{G}$ holds, then $h(\mathcal{F}_1, \mathcal{G}) > h(\mathcal{F}_2, \mathcal{G}) \Leftrightarrow h(\mathcal{F}, \mathcal{F}_1 \cap \mathcal{G}) < h(\mathcal{F}, \mathcal{F}_2 \cap \mathcal{G})$.

\item If $h(\mathcal{F}, \mathcal{G}_1) < h(\mathcal{F}, \mathcal{G}_2)$ holds, then for every $C \in \mathcal{G}_1 \cap \mathcal{G}_2$ also $h(\mathcal{F} \cup \{ C \}, \mathcal{G}_1) < h(\mathcal{F} \cup \{ C \}, \mathcal{G}_2)$ holds, provided that $\mathcal{F} \cap \mathcal{G}_1 = \mathcal{F} \cap \mathcal{G}_2$ holds.

\item If $h(\mathcal{F}_1,\mathcal{G}) < h(\mathcal{F}_2,\mathcal{G})$ and $\mathcal{G}^\prime \cap \mathcal{F}_i = \mathcal{G} \cap \mathcal{F}_i$ hold, then also $h(\mathcal{F}_1,\mathcal{G}^\prime) < h(\mathcal{F}_2,\mathcal{G}^\prime)$ holds.

\end{enumerate}

\end{definition}

Note that condition (5) in Definition \ref{def-plausible} also implies a reverse implication, i.e. $h(\mathcal{F}_1 \cup \{ C \}, \mathcal{G}) < h(\mathcal{F}_2 \cup \{ C \}, \mathcal{G})$ implies $h(\mathcal{F}_1, \mathcal{G}) \le h(\mathcal{F}_2, \mathcal{G})$.
Also note that (6) and (4) imply (2) by  substitution into (6) of filters in (2) as follows. $\mathcal{F}_1=\mathcal{F}$, $\mathcal{F}_2=\mathcal{F} \cap \mathcal{G}$, $\mathcal{G}=\mathcal{G}$, $\mathcal{F}=\top$. (6) gives
\[
h(\mathcal{F},\mathcal{G})>h(\mathcal{F} \cap \mathcal{G},\mathcal{G})\Longleftrightarrow h(\top,\mathcal{F} \cap \mathcal{G})< h(\top,\mathcal{F} \cap \mathcal{G}),
\]
thus implying $h(\mathcal{F},\mathcal{G})\le h(\mathcal{F} \cap \mathcal{G},\mathcal{G})$. The inequality in the other direction is a straightforward corollary of (4).

\begin{example}

Take a lattice with top element $A$ and direct successors $B, C, D, E$. Assume that we have $h(\{A,B,E\},\{A,B,C,D\}) < h(\{A,C,E\},\{A,B,C,D\})$. Then plausibility constraint (4) implies that $h(\{A,B,E\},\{A,B,C,D\}) < h(\{A,B,D,E\},\{A,B,C,D\})$ and $h(\{A,C,E\},\{A,B,C,D\}) \le h(\{A,C,D,E\},\{A,B,C,D\})$ hold. Furthermore, plausibility constraint (4) implies that also 
\[ h(\{A,B,D,E\},\{A,B,C,D\}) \le h(\{A,C,D,E\}, \{A,B,C,D\}) \]
must hold.

\end{example}

\subsection{Linear Inequalities}

Let $h: \mathbb{F} \times \mathbb{F} \rightarrow [0,1]$ be a human-defined function that satisfies the plausibility constraints. Assume the lattice $\mathcal{L}$ contains $n+2$ elements $C_0 ,\dots, C_{n+1}$ with top- and bottom elements $C_0$ and $C_{n+1}$, respectively. From this we will now derive a set of linear inequalities of the form $\sum\limits_{x \in U} x < \sum\limits_{x \in V} x$, where the elements in $U$ and $V$ correspond to $C_1 ,\dots, C_n$.

\begin{lemma}\label{lem-plausible-order}

Let $h: \mathbb{F} \times \mathbb{F} \rightarrow [0,1]$ be a human-defined function that satisfies the plausibility constraints. Then there exists a partial order $\prec$ on the set of terms $\Sigma_I = \{ \sum_{i \in I} x_i \mid I \subseteq \{ 1,\dots,n \} \}$ and a mapping $\Phi : \mathbb{F} \rightarrow \{ \Sigma_I \mid I \subseteq \{ 1,\dots,n \} \}$ such that the following conditions hold:

\begin{enumerate}

\item For all filters $\mathcal{F}_1$, $\mathcal{F}_2$ and $\mathcal{G}$ we have $h(\mathcal{F}_1,\mathcal{G}) < h(\mathcal{F}_2,\mathcal{G})$ implies $\Phi(\mathcal{F}_1 \cap \mathcal{G}) \prec \Phi(\mathcal{F}_2 \cap \mathcal{G})$;

\item For all filters $\mathcal{F}$, $\mathcal{G}_1$, $\mathcal{G}_2$ with $\mathcal{F} \cap \mathcal{G}_1 = \mathcal{F} \cap \mathcal{G}_2$ we have $h(\mathcal{F},\mathcal{G}_1) < h(\mathcal{F},\mathcal{G}_2)$ implies $\Phi(\mathcal{G}_2) \prec \Phi(\mathcal{G}_1)$;

\item For $J \subset I$ we have $\Sigma_J \prec \Sigma_I$;

\item For $\Sigma_J \prec \Sigma_I$ and $k\not\in J\cup I$ we also have $\Sigma_J + x_k \prec \Sigma_I +x_k$.

\end{enumerate}

\end{lemma}

\begin{proof}

Let $\mathcal{L}$ contain $n+2$ elements $C_0 ,\dots, C_{n+1}$ with top- and bottom elements $C_0$ and $C_{n+1}$, respectively. Define $\Phi(\mathcal{F}) = \sum_{C_i \in \mathcal{F} - \{ C_0, C_{n+1} \} } x_i$. 

For $\mathcal{G} = \mathcal{L}$ the inequality $h(\mathcal{F}_1,\mathcal{G})<  h(\mathcal{F}_2,\mathcal{G})$  defines the inequality $\Phi(\mathcal{F}_1) \prec \Phi(\mathcal{F}_2)$. We show that inequalities between sums of variables defined this way satisfy the requirements of the Lemma.

For $\mathcal{G} \neq \mathcal{L}$ the inequality $h(\mathcal{F}_1,\mathcal{G}) < h(\mathcal{F}_2,\mathcal{G})$ implies $h(\mathcal{F}_1 \cap \mathcal{G},\mathcal{G}) < h(\mathcal{F}_2 \cap \mathcal{G},\mathcal{G})$ due to plausibility constraint (2). Plausibility constraint (8) then gives rise to $h(\mathcal{F}_1 \cap \mathcal{G},\mathcal{L}) < h(\mathcal{F}_2 \cap \mathcal{G},\mathcal{L})$, which defines the inequality $\Phi(\mathcal{F}_1 \cap \mathcal{G}) \prec \Phi(\mathcal{F}_2 \cap \mathcal{G})$ needed.

For $\mathcal{F} = \{ C_0 \}$ the inequality $h(\mathcal{F},\mathcal{G}_2) < h(\mathcal{F},\mathcal{G}_1)$ implies using plausibility constraint (6) with $\mathcal{G} = \mathcal{L}$ that $h(\mathcal{G}_1,\mathcal{G}) < h(\mathcal{G}_2,\mathcal{G})$, so the inequality $\Phi(\mathcal{G}_1)<\Phi(\mathcal{G}_2)$ is again the same as the one defined by the right hand side $\mathcal{L}$.

Let $\mathcal{F} \neq \{ C_0 \}$  with $\mathcal{F} \cap \mathcal{G}_1 = \mathcal{F} \cap \mathcal{G}_2$ (due to plausibility constraint (1) we can ignore $\mathcal{F} = \mathcal{L}$). the inequality $h(\mathcal{F},\mathcal{G}_2) < h(\mathcal{F},\mathcal{G}_1)$ defines again $\Phi(\mathcal{G}_1) \preceq \Phi(\mathcal{G}_2)$. According to plausibility constraint (7) we also have $h(\mathcal{F} - \{ C \},\mathcal{G}_2) \le h(\mathcal{F} - \{ C \},\mathcal{G}_1)$ for $C \in \mathcal{G}_1 \cap \mathcal{G}_2$. As $\mathcal{F} \cap \mathcal{G}_i \subseteq \mathcal{G}_1 \cap \mathcal{G}_2$ holds, we obtain $h(\{ C_0 \},\mathcal{G}_2) \le h(\{ C_0 \},\mathcal{G}_1)$, so the derived inequality is again the same as for the case $\mathcal{F} = \{ C_0 \}$. This shows the claimed properties (1) and (2) of the lemma.

In order to see the claimed property (3) we fix $\mathcal{F} = \{ C_0 \}$ and exploit plausibility constraint (3) using induction on the size of $I\setminus J$.

From plausibility constraint (5) we obtain $h(\mathcal{F}_1 \cup \{ C_k \},\mathcal{G}) < h(\mathcal{F}_2 \cup \{ C_k \},\mathcal{G})$ for $C_k \notin \mathcal{F}_1 \cup \mathcal{F}_2$, which gives the claimed property (4) $\Sigma_J + x_k \prec \Sigma_I + x_k$.

Finally extend the obtained partial order $\preceq$ to a total one preserving properties (3) and (4).

\end{proof}

Note that we obtain directly a partial order for the ``worst case'', i.e. the lattice $\mathcal{L}$, in which all $C_i$ ($i=1,\dots,n$) are pairwise incomparable. Thus we could have first extended $h$ to this case, where all subsets correspond to filters, and then used the arguments in the proof.

With Lemma \ref{lem-plausible-order} we reduce the problem of finding a ranking-preserving matching measure to a problem of solving a set of linear inequalities. We will exploit the properties in this lemma for the proof of our main result in the next subsection. First we investigate a general condition for realisability.

\begin{definition}

Let $\mathcal{P}$ be a set of linear inequalities on the set of terms $\{ \sum_{i \in I} x_i \mid I \subseteq \{ 1,\dots,n \} \}$. We say that $\mathcal{P}$ is {\em realisable}, if there is a substitution $v: \{ x_1 , \dots, x_n \} \rightarrow \mathbb{R}^+$ of the variables by positive real numbers such that $\sum_{i \in I} x_i$ precedes $\sum_{j \in J} x_j$ in $\mathcal{P}$ iff $\sum_{i \in I} v(x_i) < \sum_{j \in J} v(x_j)$ holds. 

\end{definition}

As all sums are finite, it is no loss of generality to seek substitutions by rational numbers, and further using the common denominator it suffices to consider positive integers only.

 For convenience we introduce the notation $U \prec V$ for multisets $U, V$ over $\{ x_1 ,\dots, x_n \}$ to denote the inequality $\sum\limits_{x_i \in U} m_U(x_i) x_i < \sum\limits_{x_j \in V} m_V(x_j) x_j$, where $m_U$ and $m_V$ are the multiplicities for the two multisets.

\begin{theorem}\label{thm-realisable}

$\mathcal{P}$ is realisable iff there is no positive integer combination of inequalities in $\mathcal{P}$ that results in $A \prec B$ with $B \subseteq A$ as multisets, i.e. $m_B(x_i) \le m_A(x_i)$ for all $i=1,\dots,n$.

\end{theorem}

\begin{proof}

The necessity of the condition is obvious, since if $\PP$ is realizable and  $A\prec B$ is a positive integer combination of inequalities from $\PP$, then $\sum_{x_i\in A}m_A(x_i)v(x_i)<\sum_{x_j\in B}m_B(x_j)v(x_j)$ follows for the realizer substitution $v\colon X\rightarrow\R_{>0}$ that contradicts to  $m_B(x_i)\le m_A(x_i)$ for all $i=1,2,\ldots ,n$.

To prove that the condition is sufficient we use \emph{Fourier--Motzkin elimination}. That is, we do induction on the number of variables. For one variable, the statement is trivial. For the sake of convenience the inequalities are transformed into the following standard form
\[
\alpha_{i_1}x_{i_1}+\alpha_{i_2}x_{i_2}+\ldots +\alpha_{i_k}x_{i_k}-\beta_{j_1}x_{j_1}-\beta_{j_2}x_{j_2}-\ldots -\beta_{j_p}x_{j_p}<0
\]

for $\alpha_{i_t},\beta_{j_r}\in\mathbb{N}$. The condition is now that no positive integer combination of the inequalities result in an inequality with all variables having nonnegative coefficient.  
Consider variable $x_q$. Let $\PP_0$ denote the set of those inequalities that do not involve $x_q$, while $\PP_1$ denote the set of those inequalities that do  involve $x_q$.

\smallskip

\noindent
\textsc{Case 1.} $x_q$ occurs only with negative coefficients in inequalities of $\PP_1$. Let us assume that the set of inequalities $\PP_0$ is realizable, i.e., there is a substitution of the other variables that makes those inequalities valid. Then  this substitution of the other variables and a large enough value for $x_q$ will satisfy the inequalities of $\PP_1$, as well.

\smallskip

\noindent \textsc{Case 2.} $x_q$ occurs only with positive coefficients in inequalities of $\PP_1$. Let $\PP_1'$ denote the set of inequalities obtained from $\PP_1$ by removing occurrences of $x_q$ from each inequality in $\PP_1$. We claim that $\PP_0\cup\PP_1'$ satisfies the property that  no positive integer combination of inequalities of $\PP_0\cup\PP_1'$ that results in $A\prec B$ with $B\subseteq A$ as multisets, i.e., $m_B(x_i)\le m_A(x_i)$ for all $i=1,2,\ldots ,n$. Indeed, if such a combination existed, then replacing the inequalities of $\PP_1'$ by their counterparts of $\PP_1$, an  invalid positive integer linear combination of $\PP$ would be obtained.   Since $\PP_0\cup\PP_1'$ has one less variable than $\PP$, by induction hypothesis it is realizable. The inequalities in  $\PP_0\cup\PP_1'$ are all strict, so  this substitution of the other variables and a small enough value for $x_q$ will satisfy the inequalities of $\PP_1$, as well.

\smallskip

\noindent \textsc{Case 3.} $x_q$ occurs with positive and negative coefficients, as well in $\PP_1$. For every pair of inequalities such that the coefficient of $x_q$ is positive in one of them and negative in the other one we create a new inequality which is a positive integer combination of the two so that the coefficient of  $x_q$ is reduced to zero, denote the set of inequalities obtained in such way $\PP^*$. Let $\PP'=\PP_0\cup \PP^*$. It is clear that there exists no positive integer combination of inequalities of $\PP'$ that have all variables with non-negative coefficients, since any such combination is also a positive integer combination of inequalities of $\PP_0\cup \PP_1$. Thus, by the induction hypothesis $\PP'$ realizable, so there exists a substitution of the other variables than $x_q$ that makes every inequality in  $\PP'$ valid. Each inequality in $\PP_1$ that has $x_q$ with positive coefficient gives an upper bound for $x_q$ with this substitution. Also, the inequalities of $\PP_1$ containing $x_q$ with negative coefficient give lower bounds for $x_q$. We claim that the largest lower bound obtained so is smaller than the smallest upper bound. Indeed, consider the pair of inequalities giving these bounds. If the upper bound were less than or equal of the lower bound, then the positive integer combination of these two inequalities that is included in  $\PP^*$ would be violated by the given substitution of the variables  other than $x_q$.

\end{proof}

\subsection{Ranking-Preserving Matching Measures}

We now use Theorem \ref{thm-realisable} to prove the existence of ranking-preserving matching measures. As $\mathcal{P}$ is defined by $h$ we can assume that $\sum_{i \in I} x_i$ precedes $\sum_{j \in J} x_j$ for $I \subset J$. We can then also extend $\mathcal{P}$ to a partial order $\hat{\mathcal{P}}$ on multisets of variables by adding the same variable(s) to both sides. Indeed, $\Sigma_J \prec \Sigma_I\iff \Sigma_{J\setminus I} \prec \Sigma_{I\setminus J} $ by (4) of Lemma~\ref{lem-plausible-order}. Clearly, $\mathcal{P}$ is realisable iff $\hat{\mathcal{P}}$ is realisable. Here a partial order is considered realized iff the collection of strict inequalities is realized as a set of linear inequalities in the sense of  Theorem~\ref{thm-realisable}

\begin{lemma}\label{lem-plausible-ranking}

Let $\mathcal{P}$ be a partial order on the set of terms $\{ \Sigma_I \mid I \subseteq \{ 1,\dots,n \} \}$ for  $\Sigma_I = \sum_{i \in I} x_i$ such that $\Sigma_J < \Sigma_I$ holds for $J \subset I$ and $\Sigma_J + x_k < \Sigma_I +x_k$ holds, whenever $\Sigma_J < \Sigma_I$ holds and $k\not\in J\cup I$. Then $\mathcal{P}$ is realisable.

\end{lemma}

\begin{proof}

Assume that $\mathcal{P}$ is not realisable. Then according to Theorem \ref{thm-realisable} there exist inequalities $U_1 < V_1 ,\dots, U_k < V_k$ in $\mathcal{P}$ such that $V = \biguplus\limits_{i=1}^k V_i \subseteq \biguplus\limits_{i=1}^k U_i = U$ as multisets and
\[ \sum\limits_{x \in U} \sum_{j=1}^k m_{U_j}(x) x < \sum\limits_{x \in V} \sum_{j=1}^k m_{V_j}(x) x . \]

Let this system of inequalities be minimal, so each subset violates the condition in Theorem \ref{thm-realisable}. We may assume without loss of generality that $U_i\cap V_i=\emptyset$. Taking the inequalities in some order let
\[ A_i = \sum\limits_{x} \sum_{j=1}^i m_{U_j}(x) x \qquad\text{and}\qquad B_i = \sum\limits_{x} \sum_{j=1}^i m_{V_j}(x) x . \]

Then for $i<k$ there always exists some $x$ with $\sum_{j=1}^i m_{U_j}(x) < \sum_{j=1}^i m_{V_j}(x)$, while $A_i \prec B_i$. 
On the other hand $\sum_{j=1}^k m_{U_j}(x) \ge \sum_{j=1}^k m_{V_j}(x)$, while $A_k \prec B_k$.

Let $V_i^\prime$ be the multiset $B_i - A_i$, i.e. the multiset of all $x$ with $m_{B_i}(x) > m_{A_i}(x)$ such that $m_{V_i^\prime}(x) = m_{B_i}(x) - m_{A_i}(x)$ holds. Each $x \in V_i^\prime$ is a witness for the violation of the condition in Theorem \ref{thm-realisable}. In particular, we have $V_i^\prime \neq \emptyset$ for all $i<k$, but $V_k^\prime = \emptyset$.

Let $V_{i+1}^{\prime\prime} = V_i^\prime \cap V_{i+1}^\prime$ as multisets, so $m_{V_{i+1}^{\prime\prime}}(x) = \min ( m_{V_i^\prime}(x), m_{V_{i+1}^\prime}(x) )$, i.e. $x$ will at most be added to $B_i$ to give $B_{i+1}$, but not to $A_i$. In particular, $V_{i+1}^{\prime\prime} \subseteq V_i^\prime$.
Take the complement $U_{i+1}^\prime$ such that $V_i^\prime = U_{i+1}^\prime \uplus V_{i+1}^{\prime\prime}$.

As $U_i < V_i$ is in $\mathcal{P}$, we also have $U_i^\prime < V_i$ in $\mathcal{P}$ for all $i>1$ ($U_1^\prime$ is not yet defined). Indeed, using $U_i\cap V_i=\emptyset$ one can see that $U_{i+1}1^\prime = V_i^\prime \cap U_{i+1}$. 

Let $B_1^\prime = V_1 = V_1^\prime$. Then proceed inductively defining $W_i = B_i^\prime - U_{i+1}^\prime$ as well as $A_{i+1}^\prime = U_{i+1}^\prime \uplus W_i$ and $B_{i+1}^\prime = V_{i+1} \uplus W_i$, which gives $A_{i+1}^\prime < B_{i+1}^\prime$ in $\hat{\mathcal{P}}$ and $B_i^\prime = A_{i+1}^\prime$. That is, we obtain a chain 
\[ B_1^\prime \le A_2^\prime < B_2^\prime \le \dots < B_{k-1}^\prime \le A_k^\prime < B_k^\prime . \]

Complement these definitions by $U_1^\prime = B_k^\prime \cap U_1 = A_1^\prime$, and $X_0 = B_k^\prime - U_1 = X_1$. Proceed inductively defining
\begin{gather*}
C_1 = U_1^\prime \uplus X_0 \prec V_1 \uplus X_1 = B_1^\prime \uplus X_1 = D_1 \;\text{and}\; X_{i+1} = X_i - ( U_{i+1} -A_{i+1}^\prime )
\end{gather*}

This gives $C_{i+1} = A_{i+1}^\prime \uplus X_i \prec B_{i+1}^\prime \uplus X_{i+1} = D_{i+1}$
and $C_{i+1} \le D_i$. Due to this construction we also have $X_i \supseteq X_{i+1}$ for all $i$. 

Furthermore, for all elements $x$ in $X_0$ there exists a maximal $i$ with $x \in V_i - U_{i+1}$ and also $x \in A_j^\prime \cap B_j^\prime$ for all $j \ge i+1$, i.e. $x \in W_i$ and $x$ is added to both sides of $U_{i+1}^\prime < V_{i+1}$ to give the inequality $A_{i+1}^\prime < B_{i+1}^\prime$. To prove this property we proceed as follows: For $x \in V_k$ there is nothing to show. For $x \notin V_k$ we have $x \in W_{k-1}$, i.e. $x$ has been added on both sides of $U_k^\prime < V_k$ to yield $A_k^\prime < B_k^\prime$. In particular, $x \notin U_k^\prime$, but $x \in B_{k-1}^\prime$. Then either $x \in V_{k-1}$ and we are done again or $x$ has been added on both sides of $U_{k-1}^\prime < V_{k-1}$, thus proceeding this way we reach a minimal $i$ with the given property---each $x$ added to $U_2^\prime < V_2$ appears in $V_1$, so the process always stops.

On the other hand the assumed non-realisability implies $m_{B_k}(x) \le m_{A_k}(x)$, so there exists a $j$  with $x \in U_j \notin U_j^\prime$ (we must have $j<i$ for the $i$ given by the just shown property). This implies $x \in X_{j-1} - X_j$, so this (occurrence of) $x$ will not appear in $X_k$. As we can do this for all $x \in X_0$, we obtain $X_k = \emptyset$, which implies $D_k = C_1$. This defines a cycle in $\hat{\mathcal{P}}$ contradicting the fact that it is a partial order. Therefore, $\mathcal{P}$ must be realisable.

\end{proof}

\begin{example}\label{bsp-inequalities1}

To illustrate the construction in the proof take the following inequalities:
\begin{alignat*}{3}
U_1 &= \qquad & x_1 + x_2 \;&<\; x_3 + x_4 & \qquad =& V_1 \\
U_2 &= \qquad & x_2 + x_3 \;&<\; x_5 & \qquad =& V_2 \\
U_3 &= \qquad & x_4 + x_5 \;&<\; x_1 + x_3 & \qquad =& V_3 \\
U_4 &= \qquad & x_3 \;&<\; x_2 & \qquad =& V_4
\end{alignat*}

Their combination (adding up the left and right hand sides) give the following:
\[ x_1 + 2 x_2 + 2 x_3 + x_4 + x_5 < x_1 + x_2 + 2 x_3 + x_4 + x_5 , \]

i.e. multiplicities on the right are always smaller or equal as those on the left. According to Theorem \ref{thm-realisable} this means that the system is not realisable.

Taking the inequalities in the given order gives the following sums :
\begin{alignat*}{3}
A_1 &= \qquad & x_1 + x_2 &< x_3 + x_4 & \qquad =& B_1 \\
A_2 &= \qquad & x_1 + 2 x_2 + x_3 &< x_3 + x_4 + x_5 & \qquad =& B_2 \\
A_3 &= \qquad & x_1 + 2 x_2 + x_3 + x_4 + x_5 &< x_1 + 2 x_3 + x_4 + x_5 & \qquad =& B_3 \\
A_4 &= \qquad & x_1 + 2 x_2 + 2 x_3 + x_4 + x_5 &< x_1 + x_2 + 2 x_3 + x_4 + x_5 & \qquad =& B_4
\end{alignat*}

In the first three inequalities we always have at least one $x_i$ on the right hand side that has a larger multiplicity than on the left hand side, so it is a {\em witness} for the violation of the condition in Theorem \ref{thm-realisable}, i.e. we have realisability for the corresponding subsystem of inequalities.

In the proof of Theorem \ref{thm-plausible} we take these witnesses into the set $V_i^\prime$, i.e. we have:
\[
V_1^\prime = \{ x_3, x_4 \} \quad
V_2^\prime = \{ x_4, x_5 \} \quad
V_3^\prime = \{ x_3 \} \quad
V_4^\prime = \emptyset
\]

Whenever we proceed from one of these partial sums to the next some of the witnesses will no longer be witnesses---these we collect in the sets $U_i^\prime$---while others remain witnesses---these will be collected in $V_i^{\prime\prime}$. Thus, we obtain:
\begin{alignat*}{2}
U_2^\prime &=  \{ x_3 \} \qquad & \qquad V_2^{\prime\prime} &= \{ x_4 \} \\
U_3^\prime &=  \{ x_4, x_5 \} \qquad & \qquad V_3^{\prime\prime} &= \emptyset \\
U_4^\prime &=  \{ x_3 \} \qquad & \qquad V_4^{\prime\prime} &= \emptyset
\end{alignat*}

As $U_i^\prime \subseteq U_i$, i.e. we have subsets of the left hand sides of the original inequalities (for $i >1$), we can remove some of the summands and still have an inequality $U_i^\prime < V_i$ in $\mathcal{P}$:
\begin{gather*}
x_3 < x_5 \qquad
x_4 + x_5 < x_1 + x_3 \qquad
x_3 < x_2
\end{gather*}

Next we add to each of these inequalities the same sum to the left and right hand side---we denote the iequalities as $A_i^\prime < B_i^\prime$ such that the left hand side of the second inequality becomes $V_1^\prime$, and we inductively obtain a chain
\[ B_1^\prime = A_2^\prime < B_2^\prime = \dots < B_{k-1}^\prime = A_k^\prime < B_k^\prime . \]

The inequalities are as follows:
\begin{gather*}
x_3 + x_4 < x_4 + x_5 \qquad
x_4 + x_5 < x_1 + x_3 \qquad
x_1 + x_3 < x_1 + x_2
\end{gather*}

So far we only exploited the realisability of the proper subsets of inequalities, not the non-realisability of the whole set of inequalities. This is, where the left hand side of the first original inequality comes in, which we intersect with the right hand side of the last constructed inequality. In our case this gives
\[ U_1^\prime = \{ x_1 , x_2 \} . \]

As $U_1^\prime \subseteq U_1$ holds, we can again reduce the left hand side of our first inequality (in this case there is nothing to be taken away), and add an additional inequality, so we obtain now:
\begin{gather*}
x_1 + x_2 < x_3 + x_4 \quad
x_3 + x_4 < x_4 + x_5 \quad
x_4 + x_5 < x_1 + x_3 \quad
x_1 + x_3 < x_1 + x_2
\end{gather*}

This case this gives us already $X_0 = \emptyset$ and a cycle in $\mathcal{P}$ contradicting the fact that we have a linear order.

\end{example}

\begin{example}\label{bsp-inequalities2}

Let us replace the third inequality in Example \ref{bsp-inequalities1} by $x_4 + x_5 < x_1 + x_2 +x_3$. Then the combination of all four inequalities (adding up the left and right hand sides) gives the following:
\[ x_1 + 2 x_2 + 2 x_3 + x_4 + x_5 < x_1 + x_2 + 2 x_3 + x_4 + x_5 , \]

i.e. according to Theorem \ref{thm-realisable} the system is not realisable.

We take again the inequalities in the given order gives the following sums (which in the proof we denoted as $A_i \prec B_i$ for $i = 1,\dots, 4$:
\begin{gather*}
x_1 + x_2 < x_3 + x_4 \\
x_1 + 2 x_2 + x_3 < x_3 + x_4 + x_5 \\
x_1 + 2 x_2 + x_3 + x_4 + x_5 < x_1 + x_2 + 2 x_3 + x_4 + x_5 \\
x_1 + 2 x_2 + 2 x_3 + x_4 + x_5 < x_1 + 2 x_2 + 2 x_3 + x_4 + x_5
\end{gather*}

Again in the first three sums we have witnesses on the right hand side with larger multiplicities than on the left hand side, so we have realisability according to Theorem \ref{thm-realisable}. Taking these witnesses into the set $V_i^\prime$ we obtain:
\[
V_1^\prime = \{ x_3, x_4 \} \qquad
V_2^\prime = \{ x_4, x_5 \} \qquad
V_3^\prime = \{ x_3 \} \qquad
V_4^\prime = \emptyset
\]

Next we collect in the sets $U_i^\prime$ witnesses from the $(i-1)$th sum that are no longer witness in $V_i^\prime$, and in $V_i^{\prime\prime}$ we collect witnesses that retain these properties. Thus, we obtain:
\begin{alignat*}{2}
U_2^\prime &=  \{ x_3 \} \qquad & \qquad V_2^{\prime\prime} &= \{ x_4 \} \\
U_3^\prime &=  \{ x_4, x_5 \} \qquad & \qquad V_3^{\prime\prime} &= \emptyset \\
U_4^\prime &=  \{ x_3 \} \qquad & \qquad V_4^{\prime\prime} &= \emptyset
\end{alignat*}

As $U_i^\prime \subseteq U_i$, i.e. we have subsets of the left hand sides of the original inequalities (for $i >1$, which give rise to the following inequalities $U_i^\prime < V_i$ in $\mathcal{P}$ for $i=2,\dots,4$:
\begin{gather*}
x_3 < x_5 \qquad
x_4 + x_5 < x_1 + x_2 + x_3 \qquad
x_3 < x_2
\end{gather*}

Adding to each of these inequalities the same sum to the left and right hand side---we denote the iequalities as $A_i^\prime < B_i^\prime$---such that the left hand side of the $(i+1)$th inequality equals the right hand side of the $i$th inequality, we obtain a chain
\begin{gather*}
x_3 + x_4 < x_4 + x_5 \quad
x_4 + x_5 < x_1 + x_2 + x_3 \quad
x_1 + x_2 = x_3 < x_1 + 2 x_2
\end{gather*}

Now $U_1^\prime = \{ x_1, x_2 \}$ gives rise to the additional first inequality $x_1 + x_2 < x_3 + x_4$, which together with the other three does not yet define a cycle.

In this case we have $X_0 = \{ x_2 \} = X_1$ and $X_2 = X_3 = X_4 = \emptyset$, which leads to the following inequalties $C_i < D_i$ for $i=1,\dots,4$:
\begin{gather*}
x_1 + 2 x_2 < x_2 + x_3 + x_4 \quad
x_2 + x_3 + x_4 < x_4 + x_5 \\
x_4 + x_5 < x_1 + x_2 + x_3 \quad
x_1 + x_2 + x_3 < x_1 + 2 x_2
\end{gather*}

These again defines a cycle in $\hat{\mathcal{P}}$ contradicting the fact that we have a linear order.

\end{example}

From Lemma \ref{lem-plausible-ranking} together with Lemma \ref{lem-plausible-order} we immediately obtain the following main result of this section.

\begin{theorem}\label{thm-plausible}

Let $h: \mathbb{F} \times \mathbb{F} \rightarrow [0,1]$ be a human-defined function that satisfies the plausibility constraints. Then there esists a matching measure $\mu$ that is ranking-preserving with respect to $h$.

\end{theorem}

\begin{proof}

According to Lemma \ref{lem-plausible-order} $h$ defines a partial order on the set of terms $\{ \Sigma_i \mid I \subseteq \{ 1,\dots,n \} \}$, which according to Lemma \ref{lem-plausible-ranking} is realisable. Then extend a substitution $v : \{ x_1 ,\dots, x_n \} \rightarrow \mathbb{R}^+$ that gives $\sum_{i \in I} v(x_i) \le \sum_{j \in J} v(x_j)$ for $\sum_{i \in I} x_i \preceq \sum_{j \in J} x_j$ in $\mathcal{P}$ to $x_0$ and $x_{n+1}$. This defines the weighting function $w$ with $w(C_i) = \cfrac{v(x_i)}{\sum_{i=0}^{n+1} v(x_i)}$ and a corresponding matching measure $\mu$. Due to properties (1) and (2) of Lemma \ref{lem-plausible-order} $\mu$ is ranking-preserving with respect to $h$.

\end{proof}

Note that Theorem \ref{thm-plausible} only states the existence of a ranking preserving matching measure $\mu$. However, we obtain solutions for the linear inequations defined by $h$ by minimising $x_1 + \dots + x_n - 1$ under the conditions $\sum_{x_j \in V} x_j - \sum_{x_i \in U} x_i > 0$. For this linear optimisation problem the well-known {\em simplex algorithm} and matching learning approaches \cite{martinez:jikm2014} can be exploited.

\subsection{Totally Preserving Matching Measures}

While Theorem \ref{thm-plausible} guarantees the existence of a ranking-preserving matching measure $\mu$ for a human-defined function $h: \mathbb{F} \times \mathbb{F} \rightarrow [0,1]$ that satisfies the plausibility constraints, it may still be the case that not all inequalities defined by $h$ are preserved by $\mu$. Another more general problem would be to find out, if from such a human-defined function $h$ we can regain a matching measure $\mu$ on $\mathbb{F}$ that preserves all the inequalities defined by $h$.

If $w$ is a weighting function on $\mathcal{L}$, then it defines a lattice homomorphism $\hat{w} : \mathbb{V} \rightarrow [0,1]$ on the lattice of matching value terms. Let $mvt(\mathcal{F},\mathcal{G})$ denote the matching value term defined by the filters $\mathcal{F}$ and $\mathcal{G}$.

\begin{definition}\label{def-preserve}

Let $h$ be a matching function satisfying the plausibility constraints, and let $w$ be a weighting function. Then $w$ {\em totally preserves} $h$, if for all $\mathcal{F}, \mathcal{G}, \mathcal{F}^\prime, \mathcal{G}^\prime \in \mathbb{F}$ $h(\mathcal{F},\mathcal{G}) \ge h(\mathcal{F}^\prime,\mathcal{G}^\prime)$ holds iff $\hat{w}(mvt(\mathcal{F},\mathcal{G})) \ge \hat{w}(mvt(\mathcal{F}^\prime,\mathcal{G}^\prime))$ holds.

\end{definition}

The following is a counterexample to the existence of a weighting function $w$ totally preserving an arbitrary human-defined function $h: \mathbb{F} \times \mathbb{F} \rightarrow [0,1]$ that satisfies the plausibility constraints. Whether additional necessary conditions of matching measures can be claimed to define plausibility constraints, is an open problem.

\begin{example}

Let $(\mathcal{L},\le)$ be a lattice with top element $x_1$, bottom element $y$, and $n-1$ pairwise incomparable elements $x_i$ ($i=2,\dots,n$, $n \ge 4$). Filters of this lattice are of the form $\mathcal{F}_U=\{x_1\}\cup \{x_u\colon u\in U\}$ for any $U\subseteq\{2,3,\ldots ,n\}$ and $\mathcal{L}$ itself.

Let $w(x_i)=w_i$ and $w(y)=w_y$, then the lattice of matching value terms $(\mathbb{V},\le)$ consists of fractions $\frac{w_1+\sum_{a\in A}w_a}{w_1+\sum_{b\in B}w_b}$ where $A\subsetneq B\subseteq\{2,3,\ldots ,n\}$ and  $\frac{w_1+\sum_{a\in A}w_a}{w_y+\sum_{i=1}^n w_i}$. From Theorem \ref{lem-v2} we immediately obtain that $\frac{w_1+\sum_{a\in A}w_a}{w_1+\sum_{b\in B}w_b}\le \frac{w_1+\sum_{c\in C}w_c}{w_1+\sum_{d\in D}w_d}$ holds in lattice $(\mathbb{V},\le)$ iff $A\subseteq C$ and $C\setminus B=D\setminus B$.

Let us consider the filter in $(\mathbb{V},\le)$ generated by $\frac{w_1+\sum_{a\in A}w_a}{w_1+\sum_{b\in B}w_b}$ with $A=\{x_2,\ldots ,x_k\}$ and $B=\{x_2,\ldots ,x_k,x_{k+1},\ldots ,x_{\ell}\}$ for some $k+1<\ell$. We need a set of weights $w_i$, such that $\exists t$ with
\[ t\le  \cfrac{w_1+\sum_{c\in C}w_c}{w_1+\sum_{d\in D}w_d}\iff \cfrac{w_1+\sum_{a\in A}w_a}{w_1+\sum_{b\in B}w_b}\le \cfrac{w_1+\sum_{c\in C}w_c}{w_1+\sum_{d\in D}w_d} \]

in lattice $(\mathbb{V},\le)$. We claim that such a weight assignment does not exist. 

Indeed, suppose that $w_i$ is a good one, and assume without loss of generality that $w_k=\min_{i\in \{1,2,\ldots ,k\}}w_i$ and $w_{k+1}=\min_{i\in \{k+1,k+2,\ldots ,\ell\}}w_i$. Let $U=B\setminus\{k,k+1\}$ and $V=B\setminus\{k\}$. As $\cfrac{w_1+\sum_{a\in A}w_a}{w_1+\sum_{b\in B}w_b}\not\le \cfrac{w_1+\sum_{u\in U}w_u}{w_1+\sum_{v\in V}w_v}$ in lattice $(\mathbb{V},\le)$ we must have 
\[\frac{w_1+\sum_{u\in U}w_u}{w_1+\sum_{v\in V}w_v}<\frac{w_1+\sum_{a\in A}w_a}{w_1+\sum_{b\in B}w_b}\] 

numerically. This is equivalent to 
\[ \left(w_1+\sum_{u\in U}w_u\right)\left(w_1+\sum_{b\in B}w_b\right)<
\left(w_1+\sum_{a\in A}w_a\right)\left(w_1+\sum_{v\in V}w_v\right).\]

Writing in more comprehensible form the above inequality is 
\begin{gather*}
(w_1+w_2+\ldots +w_{k-1}+w_{k+2}+\ldots +w_{\ell})(w_1+w_2+\ldots+w_{\ell}) < \\
\qquad\qquad\qquad (w_1+w_2+\ldots +w_{k-1}+w_{k+1}+\ldots +w_{\ell}))(w_1+w_2+\ldots+w_k).
\end{gather*}

After simplification this is equivalent with
\[(w_1+w_2+\ldots +w_{k-1}+w_{k+2}+\ldots +w_{\ell})(w_{k+2}+\ldots +w_{\ell})<w_{k+1}w_k,\]

which contradicts that $w_{k-1}w_{k+2}\ge w_{k+1}w_k$ by the choice of indices.

\end{example}

\section{Matching Queries}\label{sec:queries}

In this section we address matching queries. According to Section \ref{sec:dl} our starting point is a knowledge base with a fixed TBox and an ABox that is determined by a set of profiles. Actually, we can consider the ABox and the profiles to define a large database. Each profile $P$ defines a filter $\mathcal{F}$ in the lattice $\mathcal{L}$ that is derived from the TBox. The purpose of matching queries is to match a profile to other profiles. This can basically two forms: 

\begin{enumerate}

\item If a profile $P_r$ describes a set of requested properties, then we are looking for those profiles that fit well to the requirements. If $\mu$ is a matching measure, this means to determine profiles $P$ such that $\mu(\mathcal{F}(P),\mathcal{F}(P_r))$ is high, where $\mathcal{F}(P)$ is the representing filter of the profile $P$ according to Definition \ref{def-profile}.

\item If a profile $P_g$ describes a set of given properties, then we are looking for those profiles to which this profile fits well, i.e. to determine profiles $P$ with high matching values $\mu(\mathcal{F}(P_g),\mathcal{F}(P))$.

\end{enumerate}

In both cases we expect an ordered answer, where the order is given by decreasing matching values.

In the recruiting area case (1) refers to a job offer of a company, which defines a requested profile, for which best matching candidates are sought. Usually, these candidates are taken from a well-defined set of applicants, where different applicants may have the same profile. Here case (2) refers to the opposite, a person searching for a suitable job.

In both cases we are usually only interested in high matching values, which motivates to look at top-$k$ queries with some user-defined positive integer $k$. We want to obtain the $k$ profiles with the highest matching values, or better---as there may be many equally-ranked profiles, which make it impossible to determine exactly the $k$ best ones---we want to get those profiles, for which there are less than $k$ better-ranked profiles. For our two cases this means to determine the results of the following two queries:
\begin{align*}
(1)\; & top_k(\mu(\cdot,\mathcal{G})) = \{ P \mid | \{ P^\prime \mid \mu(\mathcal{F}(P^\prime),\mathcal{G}) > \mu(\mathcal{F}(P),\mathcal{G}) | < k \} \qquad \text{with}\; \mathcal{G} = \mathcal{F}(P_r) \\
(2)\; & top_k(\mu(\mathcal{F},\cdot)) = \{ P \mid | \{ P^\prime \mid \mu(\mathcal{F},\mathcal{F}(P^\prime)) > \mu(\mathcal{F},\mathcal{F}(P)) | < k \} \qquad \text{with}\; \mathcal{F} = \mathcal{F}(P_g)
\end{align*}

We will investigate such top-$k$ queries in Subsection \ref{ssec:topk}. We will concentrate on case (1), as case (2) can be handled in a completely analogous way.

In addition we will address gap queries, which for a profile $P$ aim at a minimal enlargement---the difference constitutes the gap---such that $P$ will appear in the result of sufficiently many top-$k$ queries. Again the two cases above can be distinguished, which can be formalised by the following two queries using user-defined positive integers $k$ and $\ell$:

\begin{enumerate}

\item For a given profile $P$ determine a profile $P^\prime$ such that $\mathcal{F}(P^\prime) - \mathcal{F}(P)$ is minimal with respect to the condition
\[ \exists P_1 ,\dots, P_\ell . P^\prime \in top_k(\mu(\cdot,\mathcal{F}(P_i))) \qquad \text{for all} \; i = 1,\dots, \ell \]

\item For a given profile $P$ determine a profile $P^\prime$ such that $\mathcal{F}(P^\prime) - \mathcal{F}(P)$ is minimal with respect to the condition
\[ \exists P_1 ,\dots, P_\ell . P^\prime \in top_k(\mu(\mathcal{F}(P_i)),\cdot) \qquad \text{for all} \; i = 1,\dots, \ell \]

\end{enumerate} 

In the recruiting application area for a given profile of a job seeker a gap query determines additional skills that should be obtained by some training or education in order to increase chances on the job market, while a gap query for a requested job offer indicates additional incentives that should be added to make the job suitable for a larger number of highly skilled candidates. We will investigate such gap queries in Subsection \ref{ssec:gap} focusing again on case (1), while case (2) can be handled analogously.

\subsection{Top-$k$ Queries}\label{ssec:topk}

For a set of profiles $\mathbb{P}$ defined by filters in a lattice $\mathcal{L}$ we denote by $\varphi$ the conditions to be met by profiles in order to be selected, then $\mathbb{P}_\varphi$ denotes the set of profiles in $\mathbb{P}$ satisfying $\varphi$ and $P_r \in \mathbb{P}$ is a \emph{required} profile driving the selection by holding the conditions $\varphi$.

\begin{quote}

For all $P \in \mathbb{P}_\varphi$ and $P^\prime \in (\mathbb{P} - \mathbb{P}_\varphi)$, select $\lambda$ profiles, where $\lambda \geq k$, out of the set of profiles $\mathbb{P}_\varphi$ such that $P$ is selected and $P^\prime$ is not selected if $\mu(\mathcal{F}(P), \mathcal{F}(P_r)) > \mu(\mathcal{F}(P^\prime), \mathcal{F}(P_r))$ and no subset of $\mathbb{P}_\varphi$ satisfies this property.

\end{quote}

In order to obtain the best-$k$ matching profiles (either job or applicant profiles) we first need to query for filters representing those profiles. In the sequel we permit to simply write $\mu(P^\prime,P)$ instead of $\mu(\mathcal{F}(P^\prime),\mathcal{F}(P))$.

\subsubsection{Preliminaries}

Let $\mathcal{F}_r$ be a filter representing the required profile $P_r$. Then, consider $l$ filters $\mathcal{F}_1, \dots ,\mathcal{F}_\ell \in \mathbb{F}$ representing profiles that satisfy $\varphi$, i.e. they match $P_r$ with matching values above a threshold $t \in [0,1]$, i.e. $\mu(\mathcal{F}_i, \mathcal{F}_r) \geq t$ for $i=1, \dots, \ell$. In addition, every filter $\mathcal{F}_i$ represents a finite number $k_i$ of profiles $P_1^i, \dots, P_{k_i}^i$ matching $P_r$, so $\mu(P_j^i, P_r) \geq t$ for all $i=1, \dots, \ell$ and $j = 1 ,\dots, k_i$. 

Then we have $\sum_{i=1}^\ell k_i = \lambda$ and $\sum_{i=1}^{\ell - 1} k_i < k$. Furthermore, any non-selected filter $\mathcal{F}_{\ell+1}$ satisfies $\mu(\mathcal{F}_{\ell+1}, \mathcal{F}_r) < t$.

\begin{example}\label{exa:prbyfl}

Let $\mathcal{L}$ be a lattice with four elements---$\mathcal{L} = \{C_1, C_2, C_3, C_4\}$---where $C_1$ is the top element, $C_4$ is the bottom element, and $C_2, C_3$ are incomparable. This defines five filters $\mathbb{F} = \{\mathcal{F}_1, \mathcal{F}_2, \mathcal{F}_3, \mathcal{F}_4, \mathcal{F}_5 \}$, as shown in Figure \ref{fig-lattice-dexa} (a) and (b), respectively.

\begin{figure}
\noindent
	\begin{subfigure}{0.25\textwidth}
		\centering
		\includegraphics[scale=0.7]{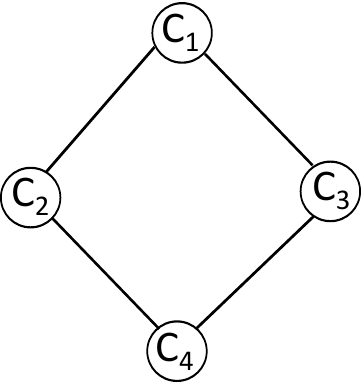} 
		\caption{Lattice}
		\label{fig:subim1}
	\end{subfigure}
	\begin{subfigure}{0.30\textwidth}\footnotesize
		\begin{align*}
			\mathcal{F}_1 &= \{C_1\} \\
			\mathcal{F}_2 &= \{C_2, C_1\} \\
			\mathcal{F}_3 &= \{C_3, C_1\} \\
			\mathcal{F}_4 &= \{C_3, C_2, C_1\} \\
			\mathcal{F}_5 &= \{C_4, C_3, C_2, C_1\} 
		\end{align*}
		\centering
		\caption{Filters}
		\label{fig:subim2}
	\end{subfigure}
	\begin{subfigure}{0.35\textwidth}\footnotesize
		{\renewcommand{\arraystretch}{1.2}
		\begin{tabular}{m{0.5cm} m{0.4cm} m{0.4cm} m{0.4cm} m{0.4cm} m{0.4cm}} 
			& $\mathcal{F}_1$	& $\mathcal{F}_2$ & $\mathcal{F}_3$	&$\mathcal{F}_4$	& $\mathcal{F}_5$\\
			\hline					
			\hline					
			$\mathcal{F}_1$ & 1	&	$\frac{1}{5}$ &	$\frac{1}{4}$	&	$\frac{1}{8}$	& $\frac{1}{10}$	\\
			\hline					
			$\mathcal{F}_2$ & 1	& 1							&	$\frac{1}{4}$							& $\frac{5}{8}$	& $\frac{1}{2}$ \\
			\hline					
			$\mathcal{F}_3$ & 1 & $\frac{4}{5}$	& 1							& $\frac{1}{2}$	& $\frac{2}{5}$ \\
			\hline					
			$\mathcal{F}_4$ & 1	& 1							& 1							& 1							& $\frac{4}{5}$ \\
			\hline					
			$\mathcal{F}_5$ & 1 & 1							& 1							& 1							& 1 \\
			\hline					
		\end{tabular}
		\centering
		\caption{Matching Measures}
		\label{fig:subim3}
		}
	\end{subfigure}
\caption{\label{fig-lattice-dexa}A lattice, its filters and corresponding matching values}

\end{figure}

If we assign the following weights to the elements of $\mathcal{L}$, $w(C_1) = \frac{1}{10}, w(C_2) = \frac{2}{5}, w(C_3) = \frac{3}{10}$, and $w(C_4) = \frac{1}{5}$ and calculate the matching values $\mu(\mathcal{F}_i, \mathcal{F}_j)$ (for $1 \leq i,j \leq 5$), we obtain the result shown in Figure \ref{fig-lattice-dexa}(c).

Assume a requested profile $A$ and four given profiles $\{B,C,D,E\}$. Let $A,B$ be represented by $\mathcal{F}_4$, $C$ by $\mathcal{F}_2$ and $D,E$ by $\mathcal{F}_3$, respectively using the filters in Figure \ref{fig-lattice-dexa}. Then we obtain the matching values $\mu(B, A) = 1$, $\mu(C, A) = \frac{5}{8}$ and $\mu(D, A) = \mu(E, A) = \frac{1}{2}$.

If we choose $k=3$, the final answer is $\{B, C, D, E\}$ with $\lambda = 4$ and with $t = \frac{1}{2}$, i.e. all profiles are in the answer to the query $top_3(\mu(\cdot,\mathcal{F}(A)))$. If we choose $k=2$, the answer is $\{B, C\}$ with $\lambda = 2$ and with $t = \frac{5}{8}$, i.e. the answer to the query $top_2(\mu(\cdot,\mathcal{F}(A)))$ contains just the profiles $B$ and $C$.

\end{example}

In order to obtain the best $l$ filters satisfying $\varphi$, we first need to know the minimum matching value representing $l$ filters. Thus, we start by selecting any $t$. If less than $\ell$ solutions are found, we decrease $t$. If more than $\ell$ solutions are found, we increase $t$. The search stops when the $\ell$ filters satisfying $\mu(\mathcal{F}_i, \mathcal{F}_r) \geq t$ for $i=1, \dots, \ell$ are found. 

With the obtained value $t$ we query for the related $k$ profiles $P_g$ where $\mu(P_g, P_r) \geq t$. This assumes to be given the matching measures between all filters in $\mathcal{L}$ and ultimately, between all profiles represented by filters. 

As seen in Example \ref{exa:prbyfl}, matching values between filters define a matrix, where columns represent the required filters $\mathcal{F}_r$ and rows represent the given filters $\mathcal{F}_g$ as depicted in Figure \ref{fig-lattice-dexa}(c). Obtaining the solutions in $top_k(\mu(\cdot,\mathcal{G}))$ or $top_k(\mu(\mathcal{F},\cdot))$ refers to one column or one row in this matrix, respectively. The process is analogous although, the perspective is different. While reading the measures from the columns provides the so called \emph{fitness} between profiles $\mu(P_g, P_r)$, the measures read from rows are the inverted measure $\mu(P_r, P_g)$. 

If we focus on columns, when querying for a particular $\mathcal{F}_r$, there would be $\mathcal{F}_i$ ($i=1, \dots, \ell$) where $\mu(\mathcal{F}_i, \mathcal{F}_r) \geq t$. We assume all elements are in total order according to the $\leq$ relation of $\mu(\mathcal{F}_i, \mathcal{F}_r)$. The advantage is that when searching for any given $\ell$ and $t$ we only need to point to the right element in the column and search for the next consecutive $\ell-1$ elements in descending order of $\mu$. The process is analogous if searching on rows.

\subsubsection{Data Structures}

We explain next how we organize profiles. We first assume an identification label $\rho_i$ representing the number $i$ of rows and, $\sigma_i$ representing the number $i$ of columns in $\mathcal{M}$ where $i > 0$.  

\begin{definition} \label{def:col}

Given a \emph{required} filter $\mathcal{F}_r$, for every element $\mu_i$ in column $\sigma_i$ in $\mathcal{M}$ representing $\mu(\mathcal{F}_{g_x}, \mathcal{F}_r)$ there is a \emph{profile record} 
\[ (\mu_i, n^>_i, n^=_i, n^<_i, \emph{next}, \emph{prev}, p) \]

describing the matching profiles $P_g$ where:

\begin{itemize}

\item $n^>_i$ denotes the number of profiles $P_g$ where $\mu (P_g, P_r) > \mu_i$,

\item $n^=_i$ denotes the number of profiles $P_g$ where $\mu (P_g, P_r) = \mu_i$,

\item $n^<_i$ denotes the number of profiles $P_g$ where $\mu (P_g, P_r) < \mu_i$,

\item \emph{next} is a reference to the next matching value in $\sigma_i$ where $\mu(\mathcal{F}_{g_{(x+1)}}, \mathcal{F}_r) \geq \mu_i$,

\item \emph{prev} is a reference to the next matching value in $\sigma_i$ where $\mu(\mathcal{F}_{g_{(x-1)}}, \mathcal{F}_r) \leq \mu_i$,

\item $p$ is a reference to a linked-list of filters matching $\mathcal{F}_r$. 

\end{itemize}

\end{definition}

The numbers $n^>_i, n^=_i, n^<_i$ are significantly important when determining the number of profiles represented by a filter without actually querying for them, i.e., if $(n^>_i + n^=_i) \geq k$ for a given $\mu_i$ we get all profiles needed.  

References \emph{Next} and \emph{Prev} make possible to track the following greater or smaller matching value by following the references.
Every profile record contains additionally a reference $p$ to the related profiles in a $\sigma_i$ column where they are organized in a transitive closure structure, ordered by the $\leq$ elements of $\mu$. 

Note that as matching values are pre-computed, the approach also works for simultaneous (or aggregated) matching with multiple matching measures.

\begin{example}\label{exa:mrecord}

Figure \ref{fig:LList} shows a representation of profile records corresponding to $\mathcal{F}_4$ (filter representing profile $A$) from Example \ref{exa:prbyfl}. Note that only the relation to filters and profiles of $\mu = 0.5$ are shown in here in order to simplify the graph.

\end{example}

\begin{figure}
\begin{center}
\includegraphics[width=\textwidth]{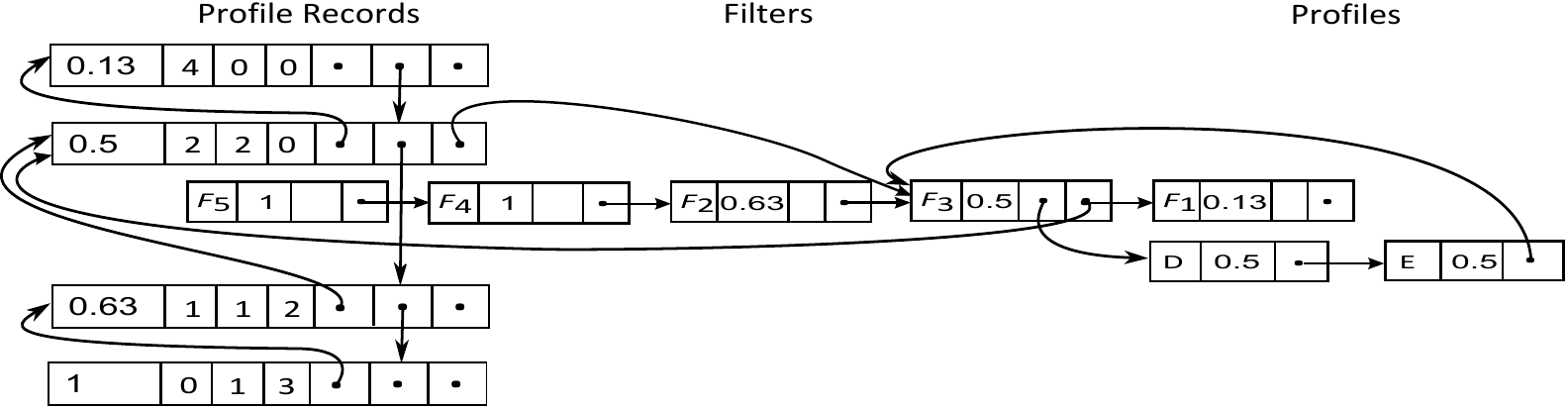} 
\caption{\label{fig:LList}Linked list of matching measures}
\end{center}
\end{figure}

Organizing data in this rings and spiders structure known from network databases leads to a better performance on the search for the top-$k$ elements. Starting by fetching a column $\sigma_i$ and together with $n^>_i$, $n^=_i$, $n^<_i$ calculate the profile records needed to get the $k$ matching profiles. Then, following the ordered linked-list structure of filters, and profiles afterward, until the $k (\lambda)$ elements are found. The definition of profile records on rows $\rho_i$ of $\mathcal{M}$ is analogous to Definition \ref{def:col}.

\subsubsection{Algorithmic Solution}

We now present an implementation of the matrix of matching values, profile records and linked-list of profiles in a relational database schema, and an algorithm that implements our definition of top-$k$ queries. Our implementation approach is designed on a relational database schema
for the storage and maintenance of filters, profiles and matching measures of an instance of $\mathcal{L}$. 

The relational schema is composed of 9 relations although, we present only two relations: \emph{ProfileRecords} and \emph{MatchingProfiles} that describe profile records as in Definition \ref{def:col} and the linked-lists of matching profiles, respectively. Figure \ref{fig:relDB} shows an example of the relations representing the elements involved in Example \ref{exa:prbyfl}.	 

\begin{figure}
\begin{center}
	\begin{subfigure}{0.50\textwidth}
		\tiny
		\begin{tabular}{|c|c|c|c|c|c|c|}
			\hline
			ID & \vtop{\hbox{\strut Required}\hbox{\strut Filter}} & Fitness & \vtop{\hbox{\strut Greater}\hbox{\strut Fitness}} & \vtop{\hbox{\strut Equal}\hbox{\strut Fitness}} & \vtop{\hbox{\strut Lesser }\hbox{\strut Fitness}} & \vtop{\hbox{\strut Next}\hbox{\strut ID}}\\
			\hline
			1	& $\mathcal{F}_4$ & 1 	 & 0  & 1  & 3 & 2 \\
			2	& $\mathcal{F}_4$ & 0.63 & 1  & 1  & 2 & 3 \\
			3	& $\mathcal{F}_4$ & 0.50 & 2  & 2  & 0 & 4 \\
			4	& $\mathcal{F}_4$ & 0.13 & 4  & 0  & 0 & null \\
			\hline
		\end{tabular}\caption{Relation \emph{ProfileRecords}}
	\end{subfigure}

	\begin{subfigure}{0.50\textwidth}
		\tiny
		\begin{tabular}{|c|c|c|c|c|c|c|}
			\hline
			ID & \vtop{\hbox{\strut Required}\hbox{\strut Filter}} & \vtop{\hbox{\strut Required}\hbox{\strut Profile}} & \vtop{\hbox{\strut Given}\hbox{\strut Filter}} & \vtop{\hbox{\strut Given}\hbox{\strut Profile}} & Fitness & \vtop{\hbox{\strut Next}\hbox{\strut ID}}\\
			\hline
			1	& $\mathcal{F}_4$ 	 & A 	  & $\mathcal{F}_4$	 & B 	& 1 		& null 	\\
			2	& $\mathcal{F}_4$ 	 & A 	  & $\mathcal{F}_2$	 & C 	& 0.63 	& null 	\\
			3	& $\mathcal{F}_4$ 	 & A 	  & $\mathcal{F}_3$	 & D 	& 0.5 	& 4 		\\
			4	& $\mathcal{F}_4$ 	 & A 	  & $\mathcal{F}_3$	 & E 	& 0.5 	& null 	\\
			\hline
		\end{tabular}\caption{Relation \emph{MatchingProfiles}}
	\end{subfigure}
\caption{Example of relations \emph{ProfileRecords} and \emph{MatchingProfiles}}\label{fig:relDB}
\end{center}
\end{figure}

For every \emph{RequiredFilter} in \emph{ProfileRecords} there is a number of matching measure, that represent the number of elements per column $\sigma_i$ in $\mathcal{M}$. The attribute \emph{NextID} in \emph{ProfileRecords} is a reference to another tuple (ID) in the relation defining the $\leq$ relation of elements of \emph{Fitness}. Attributes \emph{GreaterFitness}, \emph{EqualFitness} and \emph{LesserFitness} represent, respectively, the elements $n^>_i$, $n^=_i$, $n^<_i$ from profile records as in Definition \ref{def:col}.

In turns, for every \emph{RequiredFilter}($\mathcal{F}_r$) in \emph{ProfileRecords} there is a finite number of \emph{GivenFilters} in \emph{MatchingProfiles} that match the requirements in $\mathcal{F}_r$. The attribute \emph{NextID} is a reference to another tuple (ID) in the relation defining the $\leq$ relation of elements of \emph{Fitness}. Note that a value of \emph{null} represents the end of the list for the \emph{GivenFilter}.

Filters in a lattice $\mathcal{L}$ represent the properties of profiles via the hierarchical dependency of concepts in $\mathcal{L}$. Thus, for every \emph{required} profile $P_r$ in $\mathbb{P}$ there is a \emph{required} filter $\mathcal{F}_r \in \mathcal{L}$ representing the profile. Then retrieving the top-$k$ candidate profiles for a required filter from our relational schema is mainly performed by querying on relations \emph{ProfileRecords} and \emph{MatchingProfiles}. 

\begin{framed}\footnotesize \noindent
	\textbf{Algorithm: Top-$k$}\\
	\textbf{Input}: \\ 	
	${}$ \kern 1pc required filter: $\mathcal{F}_r$, number of matching profiles: $k$, matching threshold: $\mu$\\
	\textbf{Output}:\\
	${}$ \kern 1pc matching profiles: $P_{g_1},\dots, P_{g_k}$, measures: $\mu_{1}, \dots, \mu_{k}$\\ 
	\textbf{Begin}
	\newline
		1 \kern 1pc CREATE relation \emph{Results} = (GivenProfile, Fitness, NextID) \\
		2 \kern 1pc (fitness, count, nextid) $:= \pi_{(3,5,7)} \big( \sigma_{\underset {\emph{max}(\emph{Fitness}))}{(2=\mathcal{F}_r,}} $ (\emph{ProfileRecords}) \big) \\	
		3 \kern 1pc WHILE (count $< k$) OR (fitness $\geq \mu$) DO\\
		4 \kern 2pc  \emph{Results} $\leftarrow \pi_{(5,6,7)} \big( \sigma_{\underset{2=\mathcal{F}_r)}{(6=\text{fitness},}} (\emph{MatchingProfiles}) \big) $ \\
		5 \kern 2pc  next := \emph{Results.NextID}\\
		6 \kern 2pc  WHILE (next IS NOT NULL) DO\\
		7 \kern 3pc    \emph{Results} $\leftarrow \pi_{5,6,7} \big( \sigma_{(2=\mathcal{F}_r)} (\emph{MatchingProfiles}) \underset{1=3}{\bowtie} \emph{Results} \big)$ \\
		8 \kern 2pc  END WHILE\\
		9 \kern 2pc  (fitness, total, nextid) $:= \pi_{(3,5,7)} \big( \sigma_{(1=\text{nextid})} $ (\emph{ProfileRecords}) \big) \\
		10 \kern 2pc count := count + total\\
		11 \kern 1pc END WHILE\\
		12 \kern 1pc RETURN ($\pi_{1,2}$(\emph{Results}))\\
	\textbf{End} 
\end{framed} 

The algorithm Top-$k$ returns an ordered list of top-$k$ profiles matching a given filter. We use relational algebra notation thus, $\sigma$, $\pi$ and $\bowtie$ are the \emph{selection}, \emph{projection} and \emph{natural join} operators, respectively. Numeric subscripts are used to denote relation attributes. For instance, $\pi_1(\emph{MatchingProfiles})$ is the projection of attribute \emph{ID} of relation \emph{MatchingProfiles}.

The algorithm accepts as inputs: the required filter $\mathcal{F}_r$, the number $k$ of matching profiles and the minimum matching value $\mu$ to search for. The outputs are: the $k$ matching profiles $P_{g_1},\dots, P_{g_k}$ and their matching measures $\mu_{1}, \dots, \mu_{k}$.

With $\mathcal{F}_r$, the algorithm fetches the tuple with the greatest value of \emph{Fitness} in \emph{ProfileRecords} (line 2) and follows the references on \emph{NextID} (line 9) until the $k$ tuples are reached or $\mu(\mathcal{F}_g, \mathcal{F}_r) < t_i$ (line 3). Then, for every $\mathcal{F}_g$ in \emph{MatchingProfiles}, the algorithm queries on the linked-list of profiles (lines 6-8) and appends them in the temporary relation \emph{Results} (line 7). 
Note that by using ``fitness $\geq \mu$'' in line 3, we include the $\lambda - k$ elements as in Definition \ref{def:selection}.
The algorithm finishes by returning the elements of \emph{GivenProfile} and \emph{Fitness} on the tuples of \emph{Results}.

An implementation of B-Tree indexes on elements of \emph{Fitness} (\emph{ProfileRecords} and \emph{MatchingProfiles}) in order to access the sorted elements, as well as indexes on \emph{RequiredFilter} (\emph{ProfileRecords} and \emph{MatchingProfiles}) for random access is expected to improve performance. 
The implementation of a parallel processing on the search of matching profiles given the required profile records by calculating $(n^>_i + n^=_i)$ is another point of improvement.

\subsection{Gap Queries}\label{ssec:gap}

Let $P$ be a profile and let $\mathcal{F} = \mathcal{F}(P)$ be its representing filter in the lattice $\mathcal{L}$. Assume that a matching measure $\mu$ on $\mathcal{L}$ is fixed. We concentrate on the following case of gap queries---precisely {\em $(k,\ell)$-gap queries} for fixed positive integers $k$ and $\ell$---to determine a profile $P^\prime$ such that $\mathcal{F}(P^\prime) - \mathcal{F}$ is minimal with respect to the condition
\[ \exists P_1 ,\dots, P_\ell . P^\prime \in top_k(\mu(\cdot,\mathcal{F}(P_i))) \qquad \text{for all} \; i = 1,\dots, \ell . \]

Here and in the following minimality always refers to set inclusion. We proceed as follows:

\begin{enumerate}

\item Determine $top_\ell(\mu(\mathcal{F},\cdot)) = \{ P_1 ,\dots, P_m \}$ with $m \ge \ell$. Let $\mathcal{G}_i = \mathcal{F}(P_i)$ for $i = 1,\dots, m$ denote the filters corresponding to the profiles resulting from this top-$\ell$ query.

\item Determine $top_k(\mu(\cdot, \mathcal{G}_i)) = \{ P_1^i ,\dots, P_{k_i}^i \}$ with $k_i \ge k$ for all $i = 1,\dots,m$.

\item Let $\mathcal{D}_i^{(j)} = \mathcal{F}(P_i^j) - \mathcal{F}$ for $i = 1,\dots,m$ and $j = 1,\dots, k_i$. Each such $\mathcal{D}_i^{(j)}$ determines a gap with respect to the top-$k$ query with $\mathcal{G}_i$. As we are only interest in a minimal gap we discard all those pairs $(i,j)$, for which another $i^\prime$ exists with $\mathcal{D}_i^{(j)} \supseteq \mathcal{D}_{i^\prime}^{(j)}$.

\item Now select $\ell$ different indices $i_1 ,\dots, i_\ell$ out of $\{ 1 ,\dots, m \}$ and for each $i_x$ select $k$ different indices $j_1(i_x) ,\dots, j_k(i_x)$ out of $\{ 1 ,\dots, k_{i_x} \}$ such that $\mathcal{D}_{i_x}^{(j_y(i_x))}$ (for $x = 1 ,\dots, \ell$ and $y = 1, \dots, k$) has not been discarded in (3).

\item For each such selection of indices the union $E = \bigcup\limits_{i=1}^\ell \mathcal{D}_{i_x}^{j_h(i_x)}$ defines a {\em gap candidate}.

\end{enumerate}

Due to our construction each gap candidate $E$ satisfies the condition for the profile $P^\prime = P \cup E$ there exists at least $\ell$ profiles such such $P^\prime$ appears in the result of the top-$k$ query for all these profiles. This gives rise to the following main result.

\begin{theorem}

Each minimal gap candidate $E$ is a result for the $(k,\ell)$ gap query for profile $P$.

\end{theorem}

\section{Enriched Matching}\label{sec:probability}

In this section we extend the matching theory from Section \ref{sec:matching}. So far the developed theory is grounded in matching measures that are defined on pairs of filters in a lattice $\mathcal{L}$. The lattice $\mathcal{L}$ is derived from a TBox of a knowledge base representing the application knowledge. Thus $C_1 \le C_2$ holds in $\mathcal{L}$, if $C_1$ and $C_2$ are concepts satisfying $C_1 \sqsubseteq C_2$ in the TBox. Logically this means that the implication $C_1(x) \Rightarrow C_2(x)$ holds for all individuals $x$. However, as already mentioned in the introduction the presence of a particular concept in a profile may only partially imply the presence of another one. Therefore, we capture such additional partial implications by enriching the lattice by extra weighted edges. Then we investigate how this can be used to find the most suitable matchings. As these extra edges can introduce cycles, the notion of filter is no longer applicable. Therefore, we use a graph-based approach to extend the given and requested profiles. We show that this approach corresponds to the definition of fuzzy filters \cite{hajek:1998}. We also provide an interpretation in probabilistic logic with maximum entropy semantics \cite{beierle:entropy2015,kern-isberner:ai2004}.

\begin{figure}
\begin{center}
\includegraphics[width=122mm]{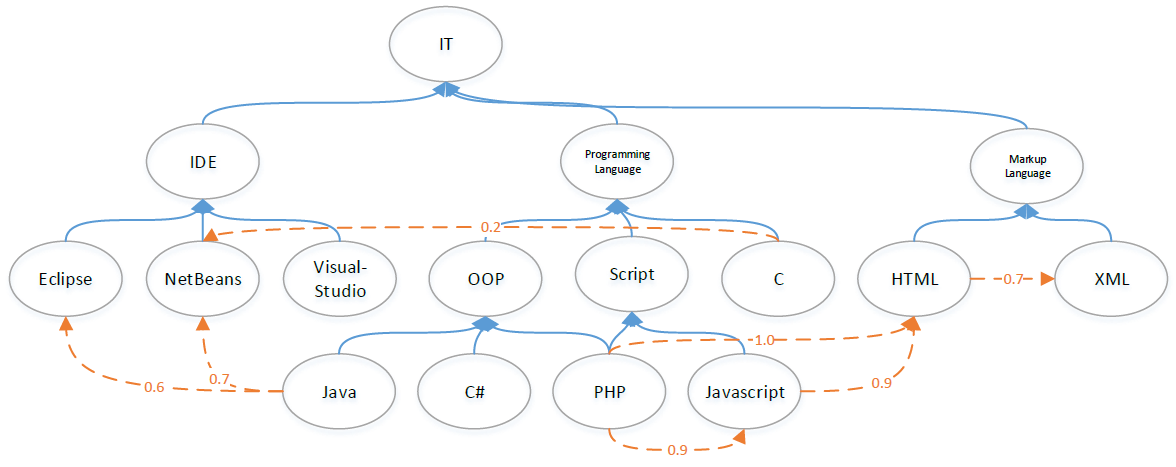}
\caption{\label{fig-graph}Fragment of a graph with lattice edges (solid) and extra edges (dashed) and assignment of degrees.}
\end{center}
\end{figure}

For the extended matching theory we could investigate again the problems of matching learning and querying that we dealt with in Sections \ref{sec:learning} and \ref{sec:queries} for the case of the filter-based matching theory. However, such investigations are still subject to ongoing research and beyond the scope of this article.

\subsection{Maximum Length Matching}

Let $G=(V, E)$ be a directed weighted graph where $V$ is a finite non-empty set of nodes and $E=\{E_O \cup E_E\} \subseteq V \times V$ is a set of edges. Each node represents a property and a directed edge $e=(v_i,v_j) \in E_O$ expresses that $v_i$ is a specialization of $v_j$ (write also $v_i \le v_j$ for this). Therefore, we claim that the subgraph $(V,E_O)$ still defines a lattice with the partial order $\le$, and we call edges in $E_O$ {\em lattice edges}. A directed edge $e=(v_i,v_j) \in E_E$ (called {\em extra edge}) represents a conditional dependency between the $v_i$ and $v_j$, namely that property $v_i$ may imply property $v_j$. For node $s, t \in V$ let $p_E(s ,t)$ denote the set of all directed paths from $s$ to $t$.

Let $d: E \rightarrow (0, 1]$ be a function with $d(e)=1$ for all lattice edges $e \in E_O$. For all extra edges $e \in E_E$ the value $d(e)$ expresses the degree (or conditional probability) to which property $v_j$ holds, when property $v_i$ is given. See Figure \ref{fig-graph} for a fragment of such a graph.  

\begin{definition}

A {\em fuzzy set} on $S$ is a mapping $f:S \rightarrow [0,1]$. A fuzzy set is called empty if $f$ is identically zero on $S$. For $t \in [0,1]$ the set $f_{t}=\{ s \in S | f(s) \geq t \}$ is called a {\em level subset} of the fuzzy set $f$. 

If $(S, \preceq)$ is a lattice and $f$ is a fuzzy set on $S$, then $f$ is called a {\em fuzzy filter}, if for all $t \in [0,1]$, $f_t$ is either empty or a filter of $S$.

\end{definition}

Note that as $S$ is a finite set, we can define a fuzzy set by enumerating all elements in $S$ with their assigned values (called the grade of membership) if that value is greater than zero, i.e. $f = \{ (s,f(s)) | s \in S \text{ and } f(s)>0\}$. For fuzzy sets $f,g$ in $S$, we define their intersection and their union as $(f \cap g)(s):=\min\{ f(s), g(s) \}$ and $(f \cup g)(s):=\max\{ f(s), g(s) \}$ for all $s \in S$, respectively.

A weighting function $w$ on a lattice $\mathcal{L}$ can be easily extended to fuzzy sets on $\mathcal{L}$ by $w(f) = \sum_{C \in \mathcal{L}} f(C) \cdot w(\{ C \})$. If $f,g$ are fuzzy filters on $\mathcal{L}$, then the matching measure $\mu$ defined by the weighting function $f$ naturally extends to $\mu(f,g) = \cfrac{w(f \cap g)}{w(g)}$.

We will now show how to derive fuzzy filters from graphs as above, i.e. lattices that are extended by extra edges. For this we extend the representing filters of a set of properties $O \subseteq V$ to fuzzy sets. The extension $\widehat{O}$ of $O$ with respect to $E$ is defined as the set of all the properties $v \in V$ that are reachable from $O$ via a directed path containing lattice and extra edges, and the grade of membership assigned to each reachable $v$ is the length of the longest path between $O$ and $v$, i.e.
\[ \widehat{O} = \{(v,\gamma_v) |  v \in V \text{ and } \exists v^\prime \in O: |p_E(v^\prime, v)| \ge 1 \text{ and } \gamma_v=\max_{v^\prime \in O, p \in p_E(v^\prime, v)} length(p) \} , \]

where $length(p) = \prod_{i=1}^{n-1} d((v_i,v_{i+1}))$ is the product of the degrees assigned to the edges on the path $p$. If the edge weights are interpreted as probabilities, then the length of longest path from $v^\prime$ to $v$ has to be interpreted as the conditional probability of property $v$ under the condition $v^\prime$.

In general, finding the longest path between two nodes in a graph is a NP-hard problem. In our case, however, the length of a path is defined as the product of the degrees of the edges on the path, so it is $exp(- L)$ with
\[ L = - \max\limits_{p \in p_E(v^\prime, v)} \log{\{\prod_{i=1}^{n-1} d((v_i,v_{i+1}))\}} = 
\min\limits_{p \in p_E(v^\prime, v)} \sum_{i=1}^{n-1} -\log{d((v_i,v_{i+1}))} . \]

Moreover, as $d(v_i,v_{i+1}) \in (0,1]$ holds for all edges, we have $-\log{d((v_i,v_{i+1}))} \ge 0$, i.e. the determination of the longest paths can be reduced to a single-source shortest path problem with non-negative edge weights, which can be solved with Dijkstra's algorithm in $O(|E| + |V|\log{|V|})$ time \cite{dijkstra:nm1959}.

\begin{theorem}

Let $(\mathcal{L}, E_O \cup E_E)$ be a directed graph with degree function $d$ extending the lattice $(\mathcal{L}, \le)$, and let $O \subseteq \mathcal{L}$ be a non-empty subset. Then the extension $\widehat{O}$ of $O$ with respect to $E$ is a fuzzy filter in $\mathcal{L}$.

\end{theorem}

\begin{proof}

For $t \in [0,1]$ we have that $\widehat{O}_t = \{ s \in \mathcal{L} \mid \widehat{O}(s) \ge t \}$ is a filter in $\mathcal{L}$ if for all $s,s' \in \mathcal{L}$ with $s \le s'$ whenever $s \in \widehat{O}_t$ holds, then also $s' \in \widehat{O}_t$ holds, i.e. if $\widehat{O}(s) \ge t$, then $\widehat{O}(s') \ge t$.

Let $s$ be in $\widehat{O}_t$ and let $p_{a,s}=(a=s_{i_1}, s_{i_2}, \dots , s_{i_k}=s)$ be a path of maximal between $O$ and $s$. We have to show that if an $s' \in \mathcal{L}$ is a generalization of $s$, i.e. $s \le s'$, then a path in $p_{a,s'}$ exists such that $length(p_{a,s'}) \ge length(p_{a,s})$.

As $s \le s'$ holds, there exists a path $p_{s,s'}=(s = s_{j_1}, s_{j_2}, ..., s_{j_l} = s')$ from $s$ to $s'$ of length $1$. If $p_{a,s}$ and $p_{s,s'}$ do not have any node in common except $s$, we can concatenate them to a path $p_{a,s'} = (a = s_{i_1}, \dots , s_{i_k},s_{j_1}, \dots s{j_l} = s')$, the length of which is $length(p_{a,s}) \cdot 1 \ge length(p_{a,s})$.

Otherwise, proceed by induction. Let $t$ be the first common node. Then take the paths $p_{a,t} = (a=s_{i_1}, \dots , s_{i_x} = t)$ and the $p_{t,s'}=(t = s_{j_y}, s_{j_{y+1}}..., s_{j_l} = s')$. We have $length(p_{a,t}) \ge length(p_{a,s})$, $d(e) \le 1$ for all $e \in E$ and $length(p_{t,s'}) = 1$. If these paths have only $t$ in common, then we can concatenate them, otherwise apply the induction hypothesis.

\end{proof}

\begin{example}

For the graph in Figure \ref{fig-graph} take the following sets of properties
\[ O_1 = \{ Java, Netbeans, XML \} \;\text{and}\; O_2 = \{ Java, PHP, Eclipse \} \]

These generate the following fuzzy filters:
\begin{align*}
\widehat{O_1} &= \{(Java, 1.0), (Netbeans, 1.0), (XML, 1.0), (OOP, 1.0), (PL, 1.0), \\
&\qquad\qquad (IT, 1.0), (IDE, 1.0), (ML, 1.0) \} \\
\text{and}\qquad \widehat{O_2} &= \{ (Java, 1.0), (PHP, 1.0), (Eclipse, 1.0), (OOP, 1.0), (PL, 1.0),\\ 
&\qquad\qquad (IT, 1.0), (Script, 1.0), (IDE, 1.0), (Netbeans, 0.7), \\
&\qquad\qquad (Javascript, 0.9), (HTML, 1.0), (ML, 1.0), (XML, 0.7) \}
\end{align*}

This gives rise to the intersection fuzzy filter
\begin{align*}
\widehat{O_1} \cap \widehat{O_2} &= \{(Java, 1.0), (OOP, 1.0), (PL, 1.0), (IT, 1.0), (IDE, 1.0), \\
&\qquad\qquad (Netbeans, 0.7), (XML, 0.7), (ML, 1.0)\}
\end{align*}

Assuming a weighting function $w$ that assigns the same weight to all elements we obtain the matching value $\mu(\widehat{O_1}, \widehat{O_2}) = \frac{7.4}{8}$.

\end{example}

\subsection{Probabilistic Matching}

The lattice and the extra edges can be handled from an information-theoretical point of view with probabilistic logic programs \cite{kern-isberner:ai2004}, or from set theoretic point of view with probabilistic models \cite{beierle:entropy2015} as well. Here we adapt the latter approach using apply the maximum entropy model to define a probabilistic matching method.

\subsubsection{Preliminaries}

Let $\Theta$ be a finite set. Let $R := \{a_1, \dots, a_l \}$ be a set of subsets of the power set $\mathcal P(\Theta)$ of $\Theta$, namely $a_i \in \mathcal P(\Theta)$, $i=1, \dots, l$. The elements of $R$ are called \emph{events}.

\begin{definition}

Let $X$ be some set. Let $\mathcal{A}$ be a subset of $\mathcal P(X)$. $\mathcal{A}$ is a {\em $\sigma$-algebra} over $X$, denoted by $\mathcal{A}(X)$, if:

\begin{itemize}

	\item $X \in \mathcal{A}$,

	\item if $Y \in \mathcal A$, then $(X \setminus Y) \in \mathcal{A}$, and

	\item if $Y_1, Y_2, ...$ is a countable collection of sets in $\mathcal{A}$, then their union $\bigcup\limits_{n=1}^{\infty}Y_n$ is in $\mathcal{A}$ as well.

\end{itemize}

\end{definition}

The set of full conjunction over $R$ is given by $\Omega := \left \{ \bigcap\limits_{i=1}^{l}e_i | e_i \in \{a_i, \neg{a_i} \} \right\} $,
where $a_i \in \mathcal P(\Theta)$, $i=1, \dots, l$, and $\neg{a_i}=\Theta \setminus a_i$.
It is well known fact that the $2^l$ elements of $\Omega$ are mutually disjoint and span the set $R$ (any $a_i$ can be expressed by a disjunction of elements of $\Omega$).
Therefore, the smallest $\sigma$-algebra $\mathcal A(R)$ that contains $R$ is identical to $\mathcal A(\Omega)$. For that reason we restrict the set of \emph{elementary events} (set of possible worlds) to $\Omega$ instead of the underlying $\Theta$.

\begin{definition}

A {\em measurable space} $(\Omega, \mathcal A)$ over a set $R:=\{a_1, \dots, a_l\}$ is defined by

\[ \Omega := \left \{ \bigcap\limits_{i=1}^{l}e_i | e_i \in \{a_i, \neg{a_i} \} \right\}  \qquad\text{and}\qquad \mathcal A = \mathcal A(\Omega) = \mathcal P(\Omega) . \]

If $(\Omega, \mathcal A)$ is a measurable space over $R$ with $\Omega = \{ \omega_1, \dots, \omega_n \}$, a {\em discrete probability measure $P$} (P-model) is an assignment of non-negative numerical values to the elements of $\Omega$, which sum up to $1$, i.e. $p_i = P(\omega_i) \ge 0$ for $i=1, \dots , n$ and $\sum a_i = 1$.

\end{definition}

The $n$-tuple $\vec{p} = (p_1, \dots, p_n)$ is called a {\em probability vector} (P-vector).
$W_{\Omega}$ (respectively $V_{\Omega}$) denotes the set of all possible P-models (P-vectors) for $(\Omega, \mathcal A)$.

\begin{definition}

Let $(\Omega, \mathcal A)$, $P \in W_{\Omega}$ $a,b \in \mathcal A, P(a)>0$, and $[l,u] \subseteq [0,1]$. A {\em sentence} in $(\Omega, \mathcal A)$ is a term $\langle P(b|a)=\delta; \delta \in [l,u] \rangle$ or $P(b|a)[l,u]$ (also use $P(a)=P(a|\Omega)$ as a shortcut), where $P(b|a)=P(a \cap b) / P(a)$ denotes the conditional probability of the event $b$ given $a$. Use $P(a)=P(a|\Omega)$ as a shortcut. The sentence is {\em true in $P \in W_{\Omega}$} iff $P(b|a) \in [l,u]$. Otherwise it is {\em false}.

\end{definition}

A sentence $P(b|a)[l,u]$ defines two inequalities, namely $P(b|a) \leq u$ (be less than the upper bound) and $P(b|a) \ge  l$ (be greater than the lower bound). These inequalities can be further transformed in the following way:
\begin{gather*}
P(b|a) \leq u \Leftrightarrow P(a \cap b) \leq u * P(a) \Leftrightarrow P(a \cap b) \leq u * (P(a \cap b) + P(a \cap \neg{b}))\\
P(b|a) \ge l \Leftrightarrow P(a \cap b) \ge l * P(a) \Leftrightarrow P(a \cap b) \ge l * (P(a \cap b) + P(a \cap \neg{b}))
\end{gather*}

Rearranging the inequalities and using the elementary probabilities $p_i, i=1,\dots, n$ we obtain

\begin{gather*}
P(b|a) \leq u \Leftrightarrow (1-u) \cdot \sum\limits_{i:w_i\in a \cap b}{p_i} + u \cdot  \sum\limits_{j:w_j \in a \cap \neg{b}}{p_j} \ge 0\\
P(b|a) \ge l \Leftrightarrow (1-l) \cdot \sum\limits_{i:w_i\in a \cap b}{p_i} - l \cdot  \sum\limits_{j:w_j \in a \cap \neg{b}}{p_j} \ge 0
\end{gather*}

Note, that if $u=1$ (respectively, $l=0$), then the first (second) inequality is always satisfied as $p_i \ge 0$.

\begin{definition}

Let $DB := \{c_1, \dots, c_m\}$ be a set of $m$ sentences in $(\Omega, \mathcal A)$. $W_{DB}$ is defined as the set of all P-models $P \in W_{\Omega}$ in which $c_1, \dots, c_m$ are true. 
We call $c_1, \dots, c_m$ {\em constraints} on $W_{\Omega}$, and $W_{DB}$ denotes the set of all elements of $W_{\Omega}$ that are consistent with the constraints in $DB$.

\end{definition}

If $W_{DB}$ is empty, the information in DB is inconsistent. If $W_{DB}$ contains more than one element, the information in DB is incomplete for determining a single P-model. 

\subsubsection{Maximum entropy model}

If $W_{DB}$ contains more than one element, the information in DB is incomplete for determining a single P-model. Therefore, we must add further constraints to the system to get a unique model. It was shown in \cite{kern-isberner:ai2004} that the maximum entropy model adds the lowest amount of additional information between single elementary probabilities to the system. Moreover, the maximum entropy model also satisfies the principle of indifference and the principle of independence:

\begin{itemize}

\item The {\em principle of indifference} states that if we have no reason to expect one event rather than another, all the possible events should be assigned the same probability.

\item The {\em principle of independence} states the if the independence of two events $a$ and $b$ in a P-model $\omega$ is given, any knowledge about the event $a$ does not change the probability of $b$ (and vice verse) in $\omega$, formally $P(b|a)=P(b)$.

\end{itemize}

In order to obtain the consistent P-model to a DB that has the maximum entropy, we have to solve the following linear optimization problem:

\begin{definition}

Let  $DB:=\{c_1, \dots, c_m\}$ be a set of $m$ sentences in $(\Omega, \mathcal A)$ with $\Omega=\{\omega_1, \dots, \omega_n\}$. Let $W_{DB}$ (respectively $V_{DB}$) be the set of all P-models $P \in W_{\Omega}$ (P-vectors $\vec{p} \in V_{\Omega}$) in which $c_1, \dots, c_m$ are true.

The maximum entropy problem is to obtain

\[ \max\limits_{\vec{v}=(p_1, ..., p_n) \in [0,1]^n}{-\sum\limits_{i=1}^n{p_i \log{p_i}}} \]

subject to the following conditions:

\begin{align*}
& p_i \ge 0 \;\text{for} i=1, \dots, n \qquad\text{and}\qquad \sum\limits_{i=1}^n{p_i} = 1 \\
& (1-u) \cdot \sum\limits_{i:w_i\in a \cap b}{p_i} + u \cdot  \sum\limits_{j:w_j \in a \cap \neg{b}}{p_j} \ge 0 \qquad\text{for all} \;\; c=P(b|a)[l,u] \in DB, l>0 \\
& (1-l) \cdot \sum\limits_{i:w_i\in a \cap b}{p_i} - l \cdot \sum\limits_{j:w_j \in a \cap \neg{b}}{p_j} \ge 0 \qquad \text{for all} \;\; c=P(b|a)[l,u] \in DB, u<0
\end{align*}

\end{definition}

In the following let $me[DB]$ denote the P-model that solve the maximum entropy problem, provided such model exists.

\begin{definition}

Let $DB$ be a set of sentences and let $c=\langle P(b|a)[l,u] \rangle$ be a sentence. We say that $c$ is a {\em maximum entropy consequence} of $DB$, denoted by $DB \parallel\sim^{me} c$, iff either $DB$ is inconsistent or $me[DB](b|a) \in [l,u]$ holds.

A {\em probabilistic query} is an expression $QP_{DB}(b|a)$ where $a$ and $b$ are two events, i.e. $a,b \in \mathcal A$, and $DB$ is a set of sentences. 

Let $DB$ be a set of sentences and let $QP(b|a)$ be a probabilistic query. The {\em answer to the query} is $\delta = me[DB](b|a)=\cfrac{me[DB](a \cap b)}{me[DB](a)}$, if $DB \parallel\sim^{me} P(a)(0,1]$. Otherwise $\delta = -1$.

\end{definition}

The query means what is the probability of $b$ given $a$ with respect to $DB$. An answer $\delta = -1$ means that the set $DB \cup \{P(a)(0,1]\}$ is inconsistent.
 
\subsubsection{Probabilistic Matching}

We now show how the matching problem can be expressed with probabilistic queries. For this let again $G=(V, \{E_O \cup E_E\})$ be a directed graph with degree function $d$. We construct $DB$ from $G$ in the following way:

\begin{enumerate}

\item Assign a new set (event) $a_v$ to each node $v$ of $G$ which contains property sets that include the property $v$.

\item For each lattice edge $(v_i,v_j) \in E_O$ add a new sentence $s_{ij}$ to the DB in the form $P(a_{v_j} | a_{v_i})[1,1]$. The sentence means if a property set contains $v_i$---i.e., is an element of $a_{v_i}$)---then it also contains $v_j$---is an element of $a_{v_j}$) as well. 

\item Then, for every extra edge  $(v_i,v_j) \in E_E$ we also add a new statement in the form of $P(a_{v_j} | a_{v_i})=[l,u]$. Here the degree of an edge can be handled in two different ways. In the first approach, let the lower bound of the interval $l$ be equal to the degree $d(v_i,v_j)$ of the edge, and let the upper bound of the interval $u$ be equal to $1$. In the second approach, let $l=u=d(v_i,v_j)$. The latter is the stricter approach, as it adds constraints to the upper bounds as well.

\end{enumerate}

A property set $A=\{v_1, ..., v_n\}$ is translated into the event $ev(A)=a_{v_1} \cap \dots \cap a_{v_n}$. 
The matching value $\mu(A_g,A_r)$ for a given profile $A_g$ with respect to a requested profile $A_r$ is the result of the probabilistic query $QP(ev(A_r)|ev(A_g))$. In this way a matching value is interpreted as the probability that the given profile is a fit for the requested one, provided that the constructed $DB$ is consistent. 

\section{Conclusions}\label{sec:schluss}

In this article we first developed a filter-based matching theory with profiles that are defined via a knowledge base. For the knowledge base we assume some TBox and an associated ABox expressed in some suitable description logic similar to DL-LITE, $\mathcal{SROIQ}$ or OWL-2, though less expressive languages might suffice. From the knowledge base we define a lattice, such that profiles give rise to filters in this lattice. This is the basis for the definition of weighted matching measures on pairs of such profiles. We showed that the theory is rather powerful, as fuzzy and probabilistic extensions can be captured.

As the theory cannot tell which weights would be appropriate for a particular application, we investigated the problem, how human-made matchings could be exploited to learn a suitable weighting function. We could show that under some mild assumptions---plausibility constraints that a human-defined matching should obey in order to exclude bias---a ranking-preserving matching measure exists, which can be used for the implementation of weight maintenance procedures. Ranking preservation takes all rankings with respect to fixed requested profiles and all rankings for fixed given profiles in connection with relevance classes into account. However, a total preservation of all relationships in the human-defined matchings is not possible. This may indicate that either the set of plausibility constraints requires extensions or the knowledge base needs to be enlarged. This open problem of adjustment of the matchings to unbiased human expertise requires further investigation. This also extends to the simultaneous capture of matchings from several human experts and the detection of probabilistic matching measures.

We further investigated the efficient implementation of matching queries, which requires to address top-$k$ queries on knowledge bases with embedded computation of matching values to evaluate the ranking criterium. This problem can be solved by pre-computed matching measures, which is efficient, as updates to the weighting functions or the TBox are considered to appear not overly often. Furthermore, in practice usually only matching values above a rather high threshold are considered to be relevant. Consequences for the system architecture of a matching system are the use of relational databases for the storage of profiles, i.e. the ABox, and the pre-computed matching values, the realisation of a querying engine on top of this database, and the coupling of updates to the matching measures with the TBox and the learning algorithms. 

For gap queries we adopted a rather pragmatic approach computing an extension of a given profile that will appear in the result of a top-$k$ matching query for at least $\ell$ requested profiles (or the other way round). With respect to the application area of job recruiting this defines a suitable extension of the profile that will improve chances of successful matching. This is of interest for targeted decisons concerning further training and education. However, we believe that gap queries require further investigation, as it should also be taken into account how likely it is that a profile can be improved, in other words how close the necessary extensions are to the properties that are already in the profile.

Prototypes of these system components will be further developed by a company towards a knowledge-based matching system in the recruitment area. Additional components comprise crawlers for job offers and candidate profiles, extraction from natural language documents, and user-friendly update interfaces for maintaining the TBox that are suitable for the use by domain experts. For other application areas, e.g. matching in sport strategies, our research results present an open invitation to explore their suitability for these applications, which may uncover further open problems for research.


\bibliographystyle{elsarticle-num} 
\bibliography{aceprom}

\end{document}